\documentclass[a4paper,11pt]{article}
\usepackage{fullpage}
\usepackage{microtype} 
\usepackage{pdfpages}
\usepackage[T1]{fontenc}
\usepackage{mathtools,amsmath,xspace}
\usepackage{amsthm}
\usepackage{amssymb}
\usepackage{tabularx}
\usepackage{multirow}
\usepackage{hyperref}
\newtheorem{theorem}{Theorem} 
\newtheorem{lemma}[theorem]{Lemma}

\newtheorem{corollary}[theorem]{Corollary}

\newtheorem{proposition}[theorem]{Proposition}

\newcommand{\homoP}{\mathcal{P}}
\newcommand{\homoPt}{\mathcal{P}_{\theta}}
\newcommand{\conv}{\mathcal{P}}
\newcommand{\ds}{\ensuremath{\mathrm{DS}(n,s)}}
\newcommand{\dsa}{\ensuremath{\mathrm{DS}(n,s,m)}}
\newcommand{\dsn}[2]{\ensuremath{\lambda_{#2}(#1)}}
\newcommand{\dsan}[3]{\ensuremath{\lambda_{#2,#3}(#1)}}

\newcommand{\ex}[2][]{\ensuremath{E_{#1}(#2)}}
\newcommand{\exf}[1]{\ensuremath{E_{#1}}}

\newcommand{\len}[2][]{\ensuremath{h_{#1}(#2)}}
\newcommand{\lenf}[1]{\ensuremath{h_{#1}}}

\newcommand{\vc}[2]{\ensuremath{\vec{s}_{#1,#2}}}
\newcommand{\pt}[2]{\ensuremath{p_{#1,#2}}}
\newcommand{\env}{\ensuremath{\mathcal{F}}}
\newcommand{\edt}{\textsf{eDT}\xspace}

\newcommand{\edtt}{\ensuremath{\textsf{eDT}_{\theta}\xspace}}

\newcommand{\aff}{\Phi}
\newcommand{\dangle}[1]{\measuredangle{#1}}
\newcommand{\ora}[1]{\overrightarrow{#1}}

\newcommand{\myparagraph}[1]{\bigskip\noindent{\textbf{#1}}}

\title{Largest similar copies of convex polygons 
amidst polygonal obstacles\thanks{This research was supported by the Institute of Information \& communications 
Technology Planning \& Evaluation(IITP) grant funded by the Korea government(MSIT)
(No. 2017-0-00905, Software Star Lab (Optimal Data Structure and Algorithmic Applications in Dynamic Geometric Environment)) and (No. 2019-0-01906, Artificial Intelligence Graduate School Program(POSTECH)).}}

\author{Taekang Eom\thanks{Department of Computer Science and Engineering, Pohang University of Science and Technology, Pohang, Korea. \texttt{\{tkeom0114,juny2400\}@postech.ac.kr}}
\and Seungjun Lee\footnotemark[2]
\and Hee-Kap Ahn\thanks{Department of Computer Science and Engineering, Graduate School of Artificial Intelligence, Pohang University of Science and Technology, Pohang, Korea. \texttt{heekap@postech.ac.kr}}
}

\begin{document}
\date{}
\maketitle

\begin{abstract}
  Given a convex polygon $P$ with $k$ vertices and a polygonal domain
  $Q$ consisting of polygonal obstacles with total size $n$ in the plane, we study
  the optimization problem of finding a largest similar copy of $P$ that can be
  placed in $Q$ without intersecting the obstacles. We improve the
  time complexity for solving the problem to
  $O(k^2n^2\dsn{k}{4}\log{n})$.  This is progress of improving the
  previously best known results by Chew and Kedem [SoCG89, CGTA93] and
  Sharir and Toledo [SoCG91, CGTA94] on the problem in more than 25
  years.
\end{abstract}

\section{Introduction}\label{se:intro}
Finding a largest object of a certain shape that can be placed in a
polygonal environment has been considered as a fundamental problem in
computational geometry. This kind of optimization problem arises in
various applications, including the metal industry where we want to find a
largest similar pattern containing no faults in a piece of material.
There is also a correspondence to motion planning
problems~\cite{o1985retraction,KedemSharir1990,DCGhandbook} and shape
matching~\cite{Fleischer1992}.

\myparagraph{Polygon placement.} In the \textit{polygon containment
  problem}, we are given a container and a fixed shape, and want to
find a largest object of the shape that can be inscribed in the
container.  There are various assumptions on the object to be placed,
the motions allowed, and the environment the object is to be placed
within.  In many cases, the container is a convex or simple polygon,
possibly with holes.  Typical shapes are squares, triangles with fixed
interior angles, and rectangles with fixed aspect ratios.  For the
motions, we may allow translation or both translation and rotation,
together with scaling.  When scaling is not allowed, the problem is to
find a copy of a given object under translation or rigid motion
that can be inscribed in a
polygon~\cite{Chazelle1983,AvnaimBoissonnat88}.
When both translation and scaling are allowed, the objective becomes
to find a largest homothetic copy of a given object that can be
inscribed in a polygonal domain~\cite{fortune1985fast,LevenSharir87}.
When rotation is allowed, together with translation and scaling, the
problem is to find a largest similar copy of a given object and it
becomes more involved; it may require to capture every change induced
by the rotation to the underlying structure maintained by the
algorithms, and therefore the complexity of the algorithms may depend
on the changes during the rotation.

In this paper, we consider the polygon containment problem under
translation, rotation, and scaling.  We aim to find a largest similar
copy of a given convex polygon $P$ with $k$ vertices that can be
inscribed in a polygonal domain $Q$ consisting of $n$ points and line
segments.  This problem has been considered
fairly well in computational geometry community for many years. See
Chapter 50 of the Handbook of Discrete and Computational
Geometry~\cite{DCGhandbook}.

The earliest result was perhaps the SoCG'89 paper by Chew and
Kedem~\cite{chew1989placing}.  They considered the problem and gave an
incremental technique for handling all the combinatorial changes to
the Delaunay triangulation of the polygonal domain $Q$ under the
distance function induced by the input polygon $P$ while $P$ is
rotating.  They gave an upper bound $O(k^4n\dsn{kn}{4})$ on the
combinatorial changes, that is, the number of \emph{critical
  orientations}, together with a deterministic
$O(k^4n\dsn{kn}{4}\log{n})$-time algorithm, where $\dsn{n}{s}$ the
length of the longest Davenport--Schinzel sequence of order $s$
including $n$ distinct symbols.  A few years later, the bound was
improved to $O(k^4n\dsn{n}{3})$ by them, and thus the running time of
the algorithm became $O(k^4n\dsn{n}{3}\log{n})$~\cite{chew1993convex}.

Toledo~\cite{Toledo91}, and Sharir and
Toledo~\cite{sharir1994extremal} studied this problem (they called
this problem \emph{the extremal polygon containment problem}) and
applied the motion-planning algorithm~\cite{KedemSharir1990} to solve
this problem.  They gave an algorithm with running time
$O(k^2n\dsn{kn}{4}\log^3(kn)\allowbreak\log\log(kn))$ that uses the
parametric search technique of Megiddo~\cite{megiddo1983applying}.

These two results, $O(k^4n\dsn{n}{3}\log{n})$ and
$O(k^2n\dsn{kn}{4}\log^3(kn)\allowbreak\log\log(kn))$, are comparable
to each other. The latter one implies a better time bound 
for large $k$ ($k>n$) while the former one implies a better time bound
for small $k$.

There was a randomized algorithm by Agarwal et
al.~\cite{agarwal1999motion} that finds a largest similar copy in
$O(kn\dsn{kn}{6}\log^3(kn)\log^2{n})$ expected time using parametric
search technique of Megiddo~\cite{megiddo1983applying}.  The problem
was also considered for special cases.  Agarwal et al.~\cite{AAS98}
considered the problem of finding largest similar copy of a given
convex $k$-gon with contained in a convex $n$-gon and gave an
$O(kn^2\log{n})$-time algorithm.  Very recently, there were results on
two variants.  Bae and Yoon~\cite{bae2020empty} gave an
$O(n^2\log{n})$-time algorithm for finding a largest empty square
among a set of $n$ points in the plane. Lee et
al.~\cite{lee2020largest} gave an algorithm for finding a largest
triangle with fixed interior angles in a simple polygon with $n$
vertices in $O(n^2\log{n})$ time.

However, no improvement to the worst-case time bounds by Chew and
Kedem, and Sharir and Toledo has been known for the polygon placement
problem.

\myparagraph{New result.}  We make progress of improving the upper
bound on the combinatorial changes considered during the rotation and
the algorithm to compute a largest similar copy.  We present an upper
bound $O(k^2n^2\dsn{k}{4})$ on the combinatorial changes, and
this directly improves the time bound for the algorithm to
$O(k^2n^2\dsn{k}{4}\log{n})$.  This improves the previously best known
results by Chew and Kedem~\cite{chew1993convex} and Sharir and
Toledo~\cite{sharir1994extremal} in more than 25 years.

Compared to the combinatorial upper bound $O(k^4n\dsn{n}{3})$ by Chew
and Kedem, our bound is smaller than their bound asymptotically for
both $k$ and $n$. Since
$\dsn{k}{4}=o(k\log^*k)$~\cite{szemeredi1974problem} and
$\dsn{n}{3}=\Theta(n\alpha(n))$, our bound is $o(k^3n^2\log^*k)$ while
their bound is $O(k^4n^2\alpha(n))$.  Therefore, our algorithm
outperforms theirs.  Compared to the time bound
$O(k^2n\dsn{kn}{4}\log^3(kn)\allowbreak\log\log(kn))$ by Sharir and
Toledo~\cite{sharir1994extremal}, our running time outperforms theirs
for both $k$ and $n$ without resorting to parametric search technique,
because
$n\dsn{k}{4}\log n=o(\dsn{kn}{4}\log^3(kn)\allowbreak\log\log(kn))$.

Thus our result improves the best result for the problem introduced in
Chapter 50 of the Handbook of Discrete and Computational
Geometry~\cite{DCGhandbook}, and the result could be a stepping stone
to closing the problem.

\myparagraph{Overview of techniques.}  We achieve the improved upper
bound by carefully analyzing the combinatorial changes in the edge
Delaunay triangulation of $Q$ (to be defined later, shortly \edt)
while rotating $P$ and by reducing the candidate size to consider for
the changes.  Following the approach of Chew and
Kedem~\cite{chew1993convex}, our strategy consists of two parts:
Counting the combinatorial changes in \edt for a constant $k$, and
then counting the combinatorial changes with respect to $k$.

\begin{enumerate}
\item In the first part, we analyze the combinatorial changes for a
  fixed $k$.  We consider a family of functions defined for each
  vertex and edge of $P$ and compute their lower envelope. Since there
  are $O(k)$ vertices and edges of $P$, we compute $O(k)$ lower
  envelopes. We show that the complexity of each lower envelope is
  $O(n)$.  Then we compute the breakpoints on the lower envelope of
  the $O(k)$ lower envelopes.  To bound the number of combinatorial
  changes in \edt, we consider a placement of a scaled copy of $P$
  such that a vertex of $Q$ and a vertex of $P$ are in contact, which
  we call a \emph{hinge}, and show that the number of breakpoints on
  the lower envelope defined for each hinge is $O(n)$.  Since there
  are $O(n)$ hinges for a constant $k$, the number of combinatorial
  changes in \edt for $\theta$ increasing from 0 to $2\pi$ is
  $O(n^2)$.
\item In the second part, we analyze the combinatorial changes to \edt
  with respect to $k$.  A combinatorial change to \edt corresponds to
  a quadruplet of pairs, each pair consisting of an element of $Q$ and
  an element of $P$ touching each other in some placement of a scaled
  copy of $P$ simultaneously.  To count the quadruplets inducing
  combinatorial changes to \edt, we consider the triplets of such
  pairs and define a function for each triplet implying the size of
  the scaled copy of $P$ defined by the triplet, satisfying the
  followings: For the lower envelope $L$ of the functions, a
  combinatorial change corresponds to an intersection of two such
  functions appearing on $L$.  That is, every combinatorial change to
  \edt occurs at a breakpoint on the lower envelope of the functions.
  So, the complexity of the lower envelope bounds the number of
  combinatorial changes that occur during the rotation of $P$.  We
  reduce the complexity bound on the lower envelope by classifying the
  combination of pairs for the quadruplets.
\end{enumerate}

While this high-level strategy may appear similar to the previous
one~\cite{chew1993convex}, there are a few major differences and
difficulties in improving the bound. In the first part, we improve the
previous bound $O(n\dsn{n}{3})$ by Chew and Kedem as follows. We
partition the family of functions to subfamilies such that the
functions in the same subfamily have the same domain length, and
therefore the complexity of their lower envelope becomes linear to the
number of functions~\cite{bae2020empty,lee2020largest}. Thus, the
total upper bound is improved to $O(n^2)$.

In the second part, instead of the quadruplets considered by Chew and
Kedem, we consider the triplets of pairs only and show that the
functions for the triplets, in their lower envelope, give us an upper
bound on the number of the combinatorial changes to \edt.  These
functions must reflect the placement of a scaled copy of $P$ that can
be inscribed in $Q$ as well as the scaling factor. We present a
function definition satisfying this requirement and show that every
combinatorial change to \edt occurs at a breakpoint on the lower
envelope of the functions. There are $O(k^3n^2)$ such functions and
two functions intersect at most four times. Thus, the complexity of
the lower envelope of the functions is $O(\lambda_6(k^3n^2))$ as the
lower envelope corresponds to a Davenport-Schinzel sequence of order
6.  To reduce the bound, we classify the functions into types based on
the combinations of pairs defining the functions and show that any two
functions belonging to the same type intersect each other less than
four times.  By applying the partition method in the first part and
the classification above on the functions, we show that the complexity
of the lower envelope becomes $O(k^2n^2\lambda_4(k))$.

Due to the limit of space, the proofs of some lemmas and corollaries
are given in Appendix. 

\section{Preliminary}
\label{sec:prelim}
\subsection{Davenport--Schinzel sequences and lower envelopes}
\label{sec:davenport}
A Davenport--Schinzel sequence is a sequence of symbols in which the
frequency of any two symbols appearing in alternation is limited.  We
call a Davenport--Schinzel sequence of order $s$ that includes $n$
distinct symbols a $\ds$-sequence, and denote by $\dsn{n}{s}$ the
length of the longest $\ds$-sequence.  For a $\ds$-sequence $U$, let
$U(j)$ denote the $j$-th entry of $U$.  We say a symbol $a_i$ of $U$
is \emph{active} at $j$ if $U(j')=a_i$ and $U(j'')=a_i$ for some
$j'\leq j$ and $j''>j$.  Let $\nu(j)$ denote the number of active
symbols of $U$ at $j$.  We say $U$ is a $\dsa$-sequence if
$\nu(j)\leq m$ for each $j$. We denote by $\dsan{n}{s}{m}$ the length
of the longest $\dsa$-sequence.

We use some properties and lemmas related to Davenport--Schinzel
sequences in analyzing algorithms.  Let
$\mathcal{F}=\{f_1,\dots,f_n\}$ be a collection of $n$
partially-defined, continuous, one-variable real-valued functions. The
points at which two functions intersect in their graphs and endpoints
of function graphs are called the \emph{breakpoints}.  If any two
functions $f_i$ and $f_j$ intersect in their graphs at most $s$ times,
the lower envelope of $\mathcal{F}$ has at most $\dsn{n}{s+2}$
breakpoints~\cite{atallah1983dynamic}.  We introduce some technical
lemmas that are used in Section~\ref{se:EDT.change}.

\begin{lemma}[Lemma 14 of~\cite{lee2020largest}]
  \label{lem:lee}
  Assume any two functions $f_i$ and $f_j$ intersect in their graphs
  at most once and each function $f_i$ has domain $D_i$ of length $d$.
  If there is a constant $c$ such that $|{\bigcup}{D_i}|=cd$, then the
  lower envelope of $\mathcal{F}$ has $O(n)$ breakpoints.
\end{lemma}
\begin{proof}
  Let $D_\mathcal{F}$ be the union of all ${D_i}$'s.
  Let $l_i$ be the point of $D_\mathcal{F}$ at distance $di$ from the leftmost
  point of $D_\mathcal{F}$, for $i=0,1,\ldots, \lceil c\rceil$. Then
  $l_j<l_{j'}$ for any two indices $0\leq j<j'\leq \lceil c\rceil$.
  Then these points $l_i$'s partition $D_\mathcal{F}$ into intervals such that
  each interval has length $d$, except the last one of length smaller
  than or equal to $d$.
  Observe that each $D_i$ intersects at most two consecutive
  intervals. 
  
  Let $L(j)=\{f_i|_{[l_{j-1},l_j]}\mid l_j\in{D_i}\}$ and
  $R(j)=\{f_i|_{[l_j,l_{j+1}]}\mid l_j\in{D_i}\}$ be the sets of
  functions defined on the intervals with $0< j < \lceil c\rceil$.
  Since any two functions in $L(j)$ intersect in their graphs at most
  once and their domains have $l_j$ as the right endpoint,
  the lower envelope of $L(j)$ forms a Davenport--Schinzel sequence of
  order 2, and therefore the lower envelope of $L(j)$ has complexity
  $O(|L(j)|)$.  Similarly, the lower envelope of $R(j)$ has complexity
  $O(|R(j)|)$.
  Since the union $U_L$ of all $L(j)$'s and the union $U_R$ of all
  $R(j)$'s have $O(n)$ functions, the lower envelope of each union has
  complexity $O(n)$.
    
  Now, consider a new partition of $D_\mathcal{F}$ obtained by slicing it at
  every breakpoint on the lower envelopes of $U_L$ and $U_R$.  Since
  there are $O(n)$ such breakpoints and the lower envelope of $\mathcal{F}$
  restricted to a component of the new partition has a constant
  complexity, the complexity of the lower envelope of $\mathcal{F}$ is $O(n)$.
\end{proof}

\begin{lemma}[Lemma 1 of~\cite{huttenlocher1992dynamic}]
  \label{lem:Huttenlocher}
  $\dsan{n}{s}{m}$ is $O(\frac{n}{m}\dsn{m}{s})$.
\end{lemma}

\begin{corollary}
  \label{cor:low_env2}
  Let $\mathcal{G}=\{g_1,\dots,g_m\}$ be a collection of $m$
  partially-defined, piecewise continuous, one-variable real-valued
  functions.  Let $n$ be the total number of continuous pieces in the
  function graphs of $\mathcal{G}$. If any pair of continuous pieces intersect
  in at most $s$ points, the lower envelope of $\mathcal{G}$ has
  $O(\frac{n}{m}\dsn{m}{s+2})$ breakpoints.
\end{corollary}
\begin{proof}
  Let $\mathcal{G'}=\{g'_1,\dots,g'_n\}$ be the set of all continuous
  pieces in $\mathcal{G}$.  Let $I_1,\dots,I_l$ be the
  interior-disjoint intervals induced by the breakpoints of lower
  envelope of $\mathcal{G'}$, such that $I_j$ lies to the left of
  $I_{j'}$ if $j<j'$ for any two indices $j$ and $j'$.  Note that the
  set of intervals $I_1,\dots,I_l$ is a partition of the domain of
  $\mathcal{G}$.  Then the lower envelope of $\mathcal{G}$ restricted
  to interval $I_j$ consists of exactly one continuous piece of
  $\mathcal{G'}$.  Let $u_j$ be the index of the continuous piece
  appearing in the lower envelope in $I_j$.  Then $U=<u_1,\dots,u_l>$
  is the sequence representing the lower envelope of $\mathcal{G}$.
  Since the number of active symbols at any position of the domain of
  $\mathcal{G}$ is at most one, $U$ is a $\mathrm{DS}(n,s+2,m)$-sequence.
  Therefore, the lower envelope of $\mathcal{G}$ has
  $O(\frac{n}{m}\dsn{m}{s+2})$ breakpoints by
  Lemma~\ref{lem:Huttenlocher}.
\end{proof}

\subsection{Edge Voronoi diagrams and edge Delaunay triangulations}
\label{sec:voronoi}
We introduce the \emph{edge Voronoi diagram} and its dual, the
\emph{edge Delaunay triangulation} (\edt), which are described
in~\cite{chew1993convex}.
The set $S$ of sites consists of the edges (open line segments) and
their endpoints in the polygonal domain $Q$ of size $n$.  For a convex
polygon $P$ with $k$ vertices, the $P$-distance from a point $p$ to a
point $q$ is $d_{P}(p,q)=\inf\{\mu\vert q\in{p+\mu P}\}$.  The Voronoi
diagram is a subdivision of the plane into regions such that the
points in the same region all have the same nearest site under the
$P$-distance. See Figure~\ref{fig:voronoi}(a).

\begin{figure}[ht]
  \begin{center}
    \includegraphics[width=.7\textwidth]{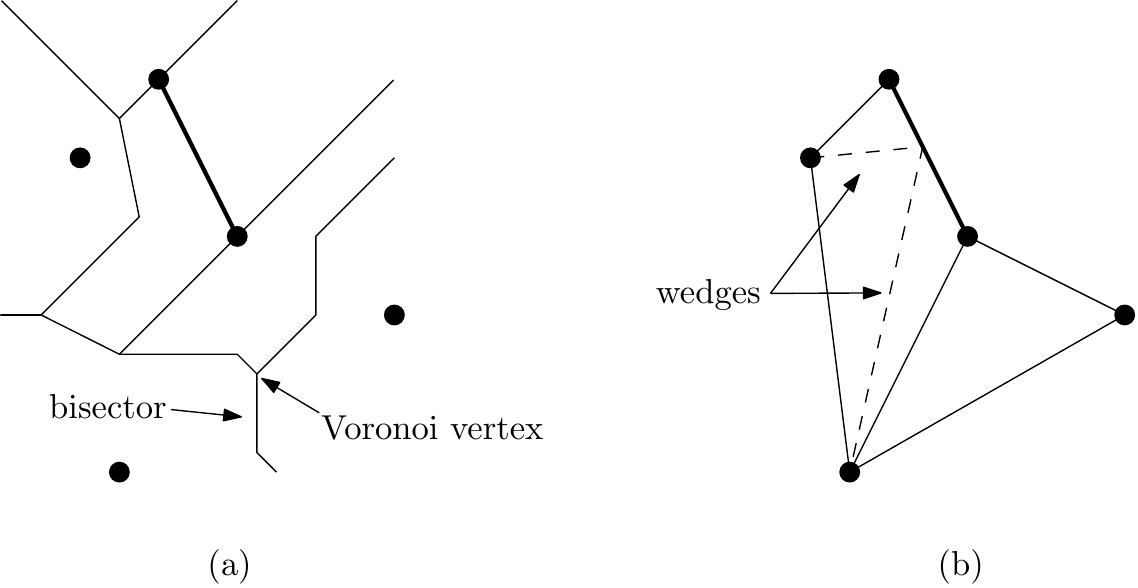}
  \end{center}
  \caption{(a) The edge Voronoi diagram of sites under $P$-distance
    when $P$ is an axis-aligned square.  (b) The edge Delaunay
    triangulation dual to the edge Voronoi diagram in (a).}
  \label{fig:voronoi}
\end{figure}

The edge Voronoi diagram consists of \emph{Voronoi vertices} and
\emph{Voronoi edges (bisectors)}.  A point in the plane is a Voronoi
vertex if and only if there is an empty circle defined by the
$P$-distance centered at the point and touching three or more sites.
It is known that the number of Voronoi vertices is linear to the
complexity of the sites~\cite{fortune1985fast}.
The bisector between two sites defines a Voronoi edge if and only if
there is an empty circle defined by the $P$-distance centered at a
point on the bisector and touching the sites.
A Voronoi edge is a polygonal line that connects two adjacent Voronoi
vertices and each point on the edge is equidistant from the two sites
defining the edge under the $P$-distance.
The edge Voronoi diagram can be built in $O(kn\log{kn})$ time and
$O(kn)$ space.

Just as the standard Delaunay triangulation is the dual of the
standard Voronoi diagram, the edge Delaunay triangulation (\edt) is
the dual of the edge Voronoi diagram.  It has three types of
\emph{generalized edges}: \emph{edges}, \emph{wedges}, and
\emph{ledges}.  An edge connects two point sites, a wedge connects a
point site and a segment site, and a ledge connects two segment
sites. See Figure~\ref{fig:voronoi}(b).  The edge Delaunay
triangulation is a planar graph consisting of point sites, segment
sites, generalized edges, and empty triangles.  Since a ledge is a
trapezoid or a degenerate trapezoid, \edt is actually not a true
triangulation in general.

The edge Delaunay triangulation can be constructed by first building
the edge Voronoi diagram and then tracing the diagram to determine the
sites that define each portion of the Voronoi edges and vertices. The
type of a generalized edge is determined by the sites defining the
corresponding Voronoi edge.

\begin{figure}[ht]
  \begin{center}
    \includegraphics[width=0.3\textwidth]{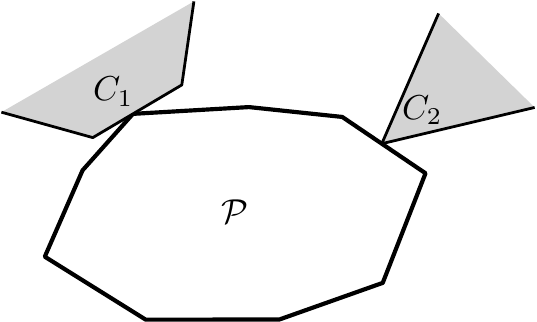}
  \end{center}
  \caption{$\homoP$ satisfies a side contact pair $C_1$ and a corner contact pair $C_2$.}
  \label{fig:fixed_angle}
\end{figure}

An ordered pair
$(A,B)$ is a \emph{side contact pair} if $A$ is a side of $Q$ and $B$
is a corner of $P$, and a \emph{corner contact pair} if $A$ is a
corner of $Q$ and $B$ is a side of $P$. A side contact pair or a corner contact
pair is called a \emph{contact pair}. We denote by $\homoP$ a homothetic copy of $P$.   We say $\homoP$
\emph{satisfies} a contact pair $(A,B)$ if $B$ in $\homoP$ touches
$A$.  See
Figure~\ref{fig:fixed_angle}. 
We say $\homoP$ is \emph{feasible} if $\homoP$ is inscribed in
$Q$.  Note that $\homoP$ is not necessarily feasible even if $\homoP$
satisfies a contact pair.

\subsection{The Algorithm of Chew and Kedem}
\label{sec:algo}
We introduce the algorithm by Chew and Kedem.  Imagine we rotate $P$
by angle $\theta$ in counterclockwise direction.  Let $P_\theta$ be
the rotated copy of $P$, and let $\homoPt$ be a homothetic copy of
$P_\theta$.  Let $\edtt$ denote the edge Delaunay triangulation of the
sites in $S$ with respect to the $P_\theta$-distance.  For a face $T$
of \edtt, we say a $\homoPt$ is associated to $T$ if it touches all
the sites defining $T$.  For $\homoPt$ associated to $T$, the set of
the elements (vertices or edges) of $\homoPt$ touching the sites
defining $T$ becomes the label of $T$. See Figure~\ref{fig:label}(a).
For a site $s$ of $S$, the label of $s$ is the set of elements of
$\homoPt$'s touching $s$, for $\homoPt$'s associated to the faces
incident to $s$. See Figure~\ref{fig:label}(b).

\begin{figure}[ht]
  \begin{center}
    \includegraphics[width=.7\textwidth]{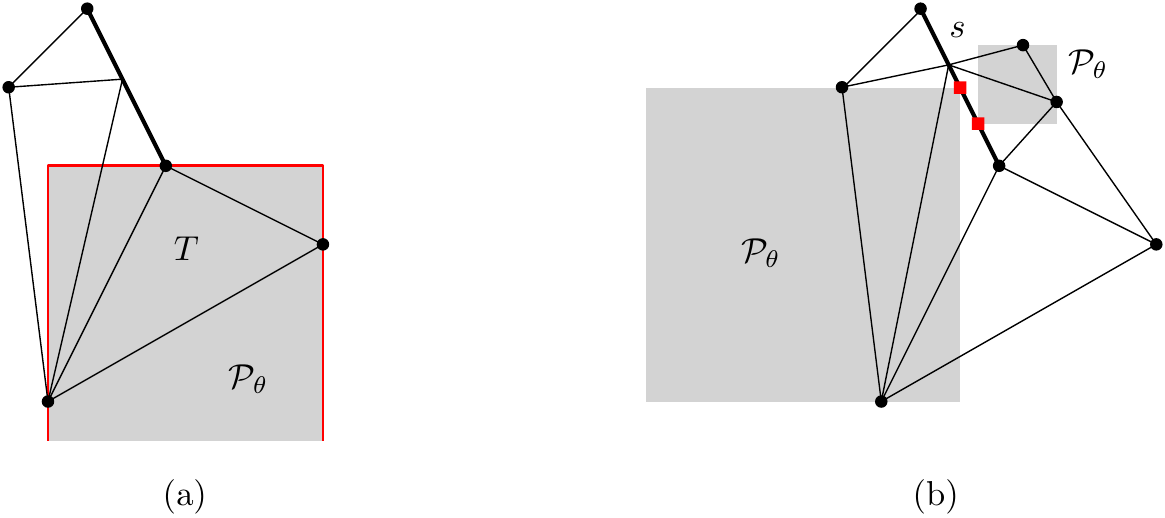}
  \end{center}
  \caption{Labels of faces of $\edtt$ and edge sites of $S$  
  for an axis-aligned square $P_\theta$. (a) The label of face $T$ is the set of red edges of
    $\homoPt$. (b) The label of edge site $s$ is the set of red
    corners of $\homoPt$'s.}
  \label{fig:label}
\end{figure}

Their algorithm classifies two possible types of changes, an edge
change and a label change in $\edtt$ while $\theta$ increases.  In an
edge change, a new generalized edge appears or an existing edge
disappears. This change occurs when $\homoPt$ touches four elements of
$Q$, resulting in a flip of the diagonals in the quadrilateral formed
by the four edges of $\edtt$.  In a label change, the label of a face
in $\edtt$ changes.  This occurs when two or more elements of
$\homoPt$ touch the same site, but the structure of $\edtt$ may be
unchanged.
An edge change or a label change is called a \emph{combinatorial
  change} to $\edtt$.

Their algorithm constructs and maintains a representation for $\edtt$
while $\theta$ increases.  It starts by creating $\edtt$ at
$\theta=0$.  For each generalized edge in $\edtt$, it determines at
which orientation this edge ceases to be valid due to an interaction
with its neighbors.  For each face in $\edtt$, the algorithm
determines at which orientation the label of this face changes.  An
edge change is detected by checking the edges of $\edtt$ and a label
change is detected by checking the faces of $\edtt$.
The algorithm maintains the edges and faces of $\edtt$ in a priority
queue, ordered by the orientations at which they are changed.  At each
succeeding stage of the algorithm, it determines which generalized
edge is the next one to disappear or which face has its label
changed as $\theta$ increases.  Then, it updates $\edtt$ and the
priority queue information for the new edge and its neighbors.  Note
that a new edge may change the priority for its neighbors. A priority
queue can be implemented such that each operation can be done
$O(\log{m})$ time, where $m$ is the maximum number of items in the
queue.  Since there are never more than $O(n)$ edges and faces in the
queue at any one time, each priority queue operation takes time
$O(\log{n})$.

For each event of a face $T$ disappearing at $\theta_e$, the algorithm
finds the maximal interval $\mathcal{I}=[\theta_s,\theta_e]$ of
$\theta$ such that $T$ appears on $\edtt$.  To find $\mathcal{I}$ in
$O(1)$ time, it stores at $T$ the orientation at which it starts to
appear on $\edtt$.  Then it computes the orientation
$\theta^* \in \mathcal{I}$ that maximize the area of each $\homoPt$
that satisfies the contacts induced by $T$.  Since $\mathcal{I}$ is
the maximal interval, $\homoPt$ is feasible for every
$\theta \in \mathcal{I}$ but not for any $\theta \notin \mathcal{I}$
sufficiently close to $\mathcal{I}$.  Thus, the algorithm considers
all orientations that $\homoPt$ is feasible and computes the placement
and orientation of the largest similar copy of $P$.  Note that the
area function of $\homoPt$ can be computed in $O(1)$ time and the
number of $\homoPt$ that satisfies the contacts induced by $T$ is
$O(1)$.  Chew and Kedem gave an upper bound $O(k^4n\dsn{n}{3})$ on the
combinatorial changes, and their algorithm takes
$O(k^4n\dsn{n}{3}\log{n})$ time~\cite{chew1993convex}.

\section{The number of changes in \texorpdfstring{$\edtt$}{eDTtheta}}
\label{se:EDT.change}
We show that the number of combinatorial changes during the rotation
is $O(k^2n^2\dsn{k}{4})$ in this section. 
This directly improves the time bound for the algorithm to 
$O(k^2n^2\dsn{k}{4}\log{n})$.
We analyze the number of combinatorial changes in $\edtt$ for a
constant $k$ in Section~\ref{sec:phase1}, and then analyze the number
of combinatorial changes with respect to $k$ in
Section~\ref{sec:phase2} using the observation in
Section~\ref{sec:phase1}.  We use the reference point $(x,y)$, the
reference orientation $\theta$, and the expansion factor $\delta$ of
$\homoPt$ defined in~\cite{chew1993convex}, which represent a
placement of $\homoPt$ in the plane by a quadruplet
$(x,y,\theta,\delta)$.

\subsection{The number of changes for fixed \texorpdfstring{$k$}{k}}
\label{sec:phase1}

\begin{figure}[ht]
  \begin{center}
    \includegraphics[width=.8\textwidth]{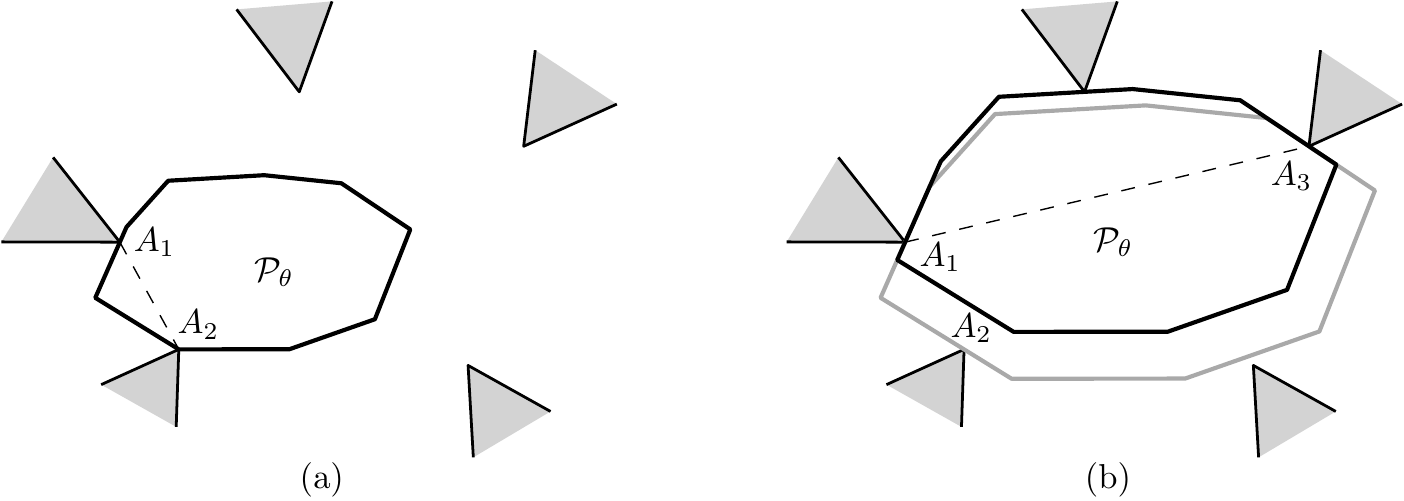}
  \end{center}
  \caption{(a) The segment $A_1A_2$ is a reported edge with $Q_H=A_2$.
    (b) The segment $A_1A_3$ is an unreported edge because no hinge is
    involved in the segment for a feasible $\homoPt$.  }
  \label{fig:reported_edge}
\end{figure}
For a constant $k$, we improve the previously best upper bound
$O(n\dsn{n}{3})$ by Chew and Kedem~\cite{chew1993convex} to
$O(n^2)$. A key idea is to consider the lower envelope of some
functions related to the expansion factor, one for each edge and
vertex of $P$, and then to analyze the lower envelope of those lower
envelopes.  By careful analysis on the complexities of the lower
envelopes, we show that the number of combinatorial changes to $\edtt$
for $\theta$ increasing from 0 to $2\pi$ is $O(n^2)$.

An ordered pair $(Q_H,P_H)$ is a \emph{hinge} if $Q_H$ is a corner of
$Q$ and $P_H$ is a corner of $P$.  For a hinge $H=(Q_H,P_H)$ and a
contact pair $C=(A,B)$, the generalized edge connecting $Q_H$ and $A$ 
is a \emph{reported edge}
if there is a feasible $\homoPt$ for some $\theta$ satisfying both $H$
and $C$.  An edge of $\edtt$ is an \emph{unreported edge} if it is not
a reported edge.  See Figure~\ref{fig:reported_edge}.

\myparagraph{Changes to the reported edges and the label changes to point sites.}
We count the changes to the reported edges and the changes to the labels 
of point sites in $\edtt$ for $\theta$ increasing from $0$ to $2\pi$.
We define the \emph{expansion function} $\ex[HC]{\theta}$ for a hinge
$H$ and a contact $C$ to be the minimal expansion factor of $\homoPt$
satisfying $H$ and $C$.  For a hinge $H=(Q_H, P_H)$, let $\env_H$ be
the set of all expansion functions satisfying $H$ and another contact
pair.  An expansion function $\ex[HC]{\theta}$ of $\env_H$ for a
contact $C=(A,B)$ appears on the lower envelope of $\env_H$ at
$\theta$ if the generalized edge connecting $Q_H$ and $A$
is a reported edge in $\edtt$.  Then the number of
changes to the reported edges in $\edtt$ which involve $H$ is bounded
by the number of breakpoints on the lower envelope of $\env_H$.

Every label change to a point site involves a hinge.  See
Figure~\ref{fig:change_label}(a).  An intersection of $\exf{HC_1}$ and
$\exf{HC_2}$ of $\env_H$ for contact pairs $C_1$ and $C_2$ appears on the
lower envelope of $\env_H$ if a label change to a point site is
induced by $C_1, C_2$ and $H$.  Then the number of label changes
to the point sites in $\edtt$ which involve $H$ is bounded by the
number of breakpoints on the lower envelope of $\env_H$.
\begin{proposition}[Proposition 3 of~\cite{chew1993convex}]
  \label{prop:intersection}
  Two expansion functions $E_{HC_1}$ and $E_{HC_2}$ intersect each
  other in at most one point in their graphs if both $C_1$ and $C_2$
  are corner contact pairs, or both are side contact pairs. If one is
  a corner contact pair and the other is a side contact pair,
  $E_{HC_1}$ and $E_{HC_2}$ intersect in at most two points in their
  graphs.
\end{proposition}

Let $v_i$ and $e_i$ denote the vertices and edges of $P$ for $i=1,\dots,k$.
For each $i=1\dots k$, let
$\mathcal{C}_{1i}=\{(e,v_i)\mid \text{$e$ is an edge of $Q$}\}$ be
the set of side contact pairs and let
$\mathcal{C}_{2i}=\{(v,e_i)\mid \text{$v$ is a vertex of $Q$}\}$ be
the set of corner contact pairs.
Let $\env_{ji}=\{\exf{HC}\mid C\in\mathcal{C}_{ji}\}$ for $j=1,2$.

\begin{lemma}
  \label{lem:breakpoints}
  The number of breakpoints on the lower envelope of $\env_{ji}$ 
  is $O(n)$ for each $j=1,2$ and $i=1,\dots,k$.
\end{lemma}
\begin{proof}
  The number of breakpoints on the lower envelope of $\env_{1i} $ is
  $O(n)$ for each $i$ since $\exf{HC_1}$ and $\exf{HC_2}$ intersect
  only at the boundaries of their intervals for
  $C_1,C_2\in\mathcal{C}_{1i}$.

  Consider now the number of breakpoints on the lower envelope of
  $\env_{2i}$. Two expansion functions $\exf{HC_1}$ and $\exf{HC_2}$
  intersect in at most one point for
  $C_1,C_2\in\mathcal{C}_{2i}$. Also, $\exf{HC}$ has the same length
  of domain for all $C\in\mathcal{C}_{2i}$.  Thus, the number of
  breakpoints on the lower envelope of $\env_{2i}$ is $O(n)$ by
  Lemma~\ref{lem:lee}.
\end{proof}

From Corollary~\ref{cor:low_env2}, Proposition~\ref{prop:intersection}, and Lemma~\ref{lem:breakpoints}, we achieve an upper bound on the
number of breakpoints on the lower envelope of $\env_H$.
\begin{lemma}
  \label{lem:low_env_hinged}
  The number of breakpoints on the lower envelope of $\env_H$ is
  $O(\dsn{k}{3}n)$.
\end{lemma}
\begin{proof}
  Let $\env_{j}=\{f_{j1},\dots,f_{jk}\}$ for $j=1,2$, where
  $f_{ji}$ is the lower envelope of $\env_{ji}$ for each
  $i=1,\dots,k$ and $j=1,2$. Let $\mathcal{L}_j$ denote the lower
  envelope of $\env_{j}$.  Then, the lower envelope of $\env_H$ is the
  lower envelope of $\mathcal{L}_{1}$ and $\mathcal{L}_{2}$.  The
  number of breakpoints on $\mathcal{L}_{j}$ is $O(\dsn{k}{3}n)$ for
  $j=1,2$ by Corollary~\ref{cor:low_env2}, Proposition~\ref{prop:intersection},
  and Lemma~\ref{lem:breakpoints}.  The number of breakpoints on
  the lower envelope of $\env_H$ is also $O(\dsn{k}{3}n)$, because two
  continuous pieces, one from $\mathcal{L}_{1}$ and one from
  $\mathcal{L}_{2}$, intersect at most two points by
  Proposition~\ref{prop:intersection}.
\end{proof}

By Lemma~\ref{lem:low_env_hinged}, the number of changes to the
reported edges and the number of label changes to the point sites are
$O(k\dsn{k}{3}n^2)$.
\begin{figure}[ht]
  \begin{center}
    \includegraphics[width=0.8\textwidth]{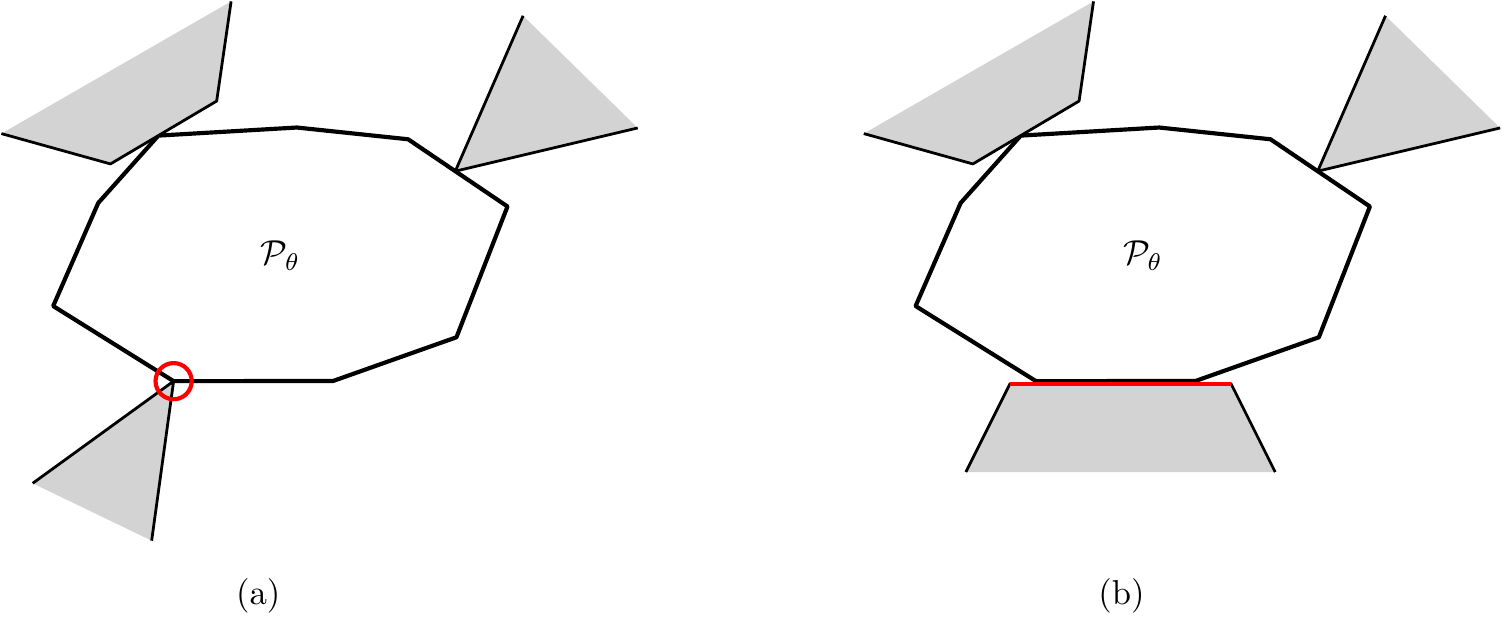}
  \end{center}
  \caption{The label changes to the red sites.  (a) Label change to a
    point site (hinge). (b) Label change to an edge site.}
  \label{fig:change_label}
\end{figure}

\myparagraph{Label changes to edge sites.}
We count the changes to the labels of edge sites in $\edtt$ for
$\theta$ increasing from 0 to $2\pi$.
Imagine we fix an edge $e$ of $Q$ and an edge $g$ of $P$. See
Figure~\ref{fig:change_label}(b).  Then, the number of label changes
to edge site $e$ with $g$ is $O(n)$ because the edge site and the edge
of $P$ are aligned.
The number of label changes to all edge sites is $O(kn^2)$.

\myparagraph{Changes to unreported edges.}
We count the changes to unreported edges using the numbers of changes
to the reported edges and to the labels, and Lemma~\ref{lem:reported_edge}.  
\begin{lemma}[Lemma 2 of~\cite{chew1993convex}]
  \label{lem:reported_edge}
  Every edge of $\edtt$ is either a reported edge or a diagonal in a
  convex $l$-gon, $l\leq3k$, whose sides are either reported edges or
  portions of edge sites.
\end{lemma}

Let $G_\theta$ be the graph whose
edges are the reported edges in $\edtt$ and portions of edge sites in
Lemma~\ref{lem:reported_edge}. We count the changes to the
unreported edges which are diagonals in a face of $G_\theta$ for an
interval of $\theta$ with no label change to $\edtt$.  Observe that no
combinatorial change occurs to $G_\theta$ for the interval.  Any
change to an unreported edge involves four sites lying on a face boundary of
$G_\theta$.  There are at most four changes for a group of four
sites. We describe the details on this bound in
Section~\ref{sec:intersection}.  Since each face has at most $3k$
edges by Lemma~\ref{lem:reported_edge}, there are at most
$\binom{3k}{4}$ such groups.  Since $O(k^4)$ changes occur to the
unreported edges for the boundary of a face $g$ of $G_\theta$ during an interval of
$\theta$ with no label change to the faces of $\edtt$ intersecting $g$,
there are $O(k^5n^2\lambda_3(k))$ combinatorial changes to $\edtt$.
\begin{theorem} \label{thm:size_constant}
  For a polygonal domain $Q$ of size $n$ and a convex $k$-gon
  $P$, the number of combinatorial changes to $\edtt$ for $\theta$
  increasing from 0 to $2\pi$ is $O(n^2)$ for a constant $k$.
\end{theorem}
\subsection{The number of changes with respect to
  \texorpdfstring{$k$}{k}}
\label{sec:phase2}
We now consider $k$ as a variable and bound the number of changes to
$\edtt$.
Since each triangular face in $\edtt$ is defined by three elements
(edges or vertices) of $P$, we choose three elements of $P$ then use
their convex hull in the counting.  Then the number of faces in $\edtt$ for all these
convex hulls is at most $O(k^3n^2)$ for $\theta$ increasing from 0 to
$2\pi$, by Theorem~\ref{thm:size_constant}.  Let $\mathcal{T}$ be the set of all faces of $\edtt$ for the convex hull $P'$ for three elements
$B_1,B_2,B_3$ of $P$ such that the contact pairs inducing the face
have $B_1,B_2$, and $B_3$ as their elements.

Consider two faces $T$ and $T'$ of $\edtt$ for two distinct
orientations $\theta=\theta_1$ and $\theta=\theta_2$ with
$\theta_1<\theta_2$ that are defined by the same sites.  We consider
$T$ and $T'$ as distinct faces if there is any change to $T$ or $T'$
in $\edtt$ for $\theta$ from $\theta_1$ to $\theta_2$.  For a face
$T\in\mathcal{T}$, let $C(T)$ be the set of contact pairs which
defines $T$, and let $I(T)$ be the interval of $\theta$ at which $T$
appears to $\edtt$.

\begin{figure}[ht]
  \begin{center}
    \includegraphics[width=.35\textwidth]{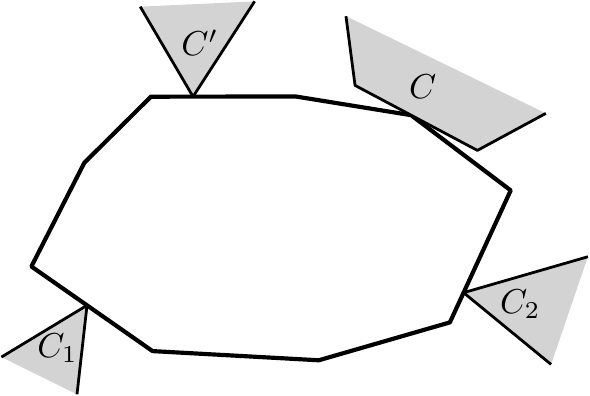}
  \end{center}
    \caption{
      The combinatorial change induced by $C_1,C_2,C$, and $C'$.}
    \label{fig:combinatorial_change}
  \end{figure}

  \begin{figure}[ht]
    \begin{center}
      \includegraphics[width=0.8\textwidth]{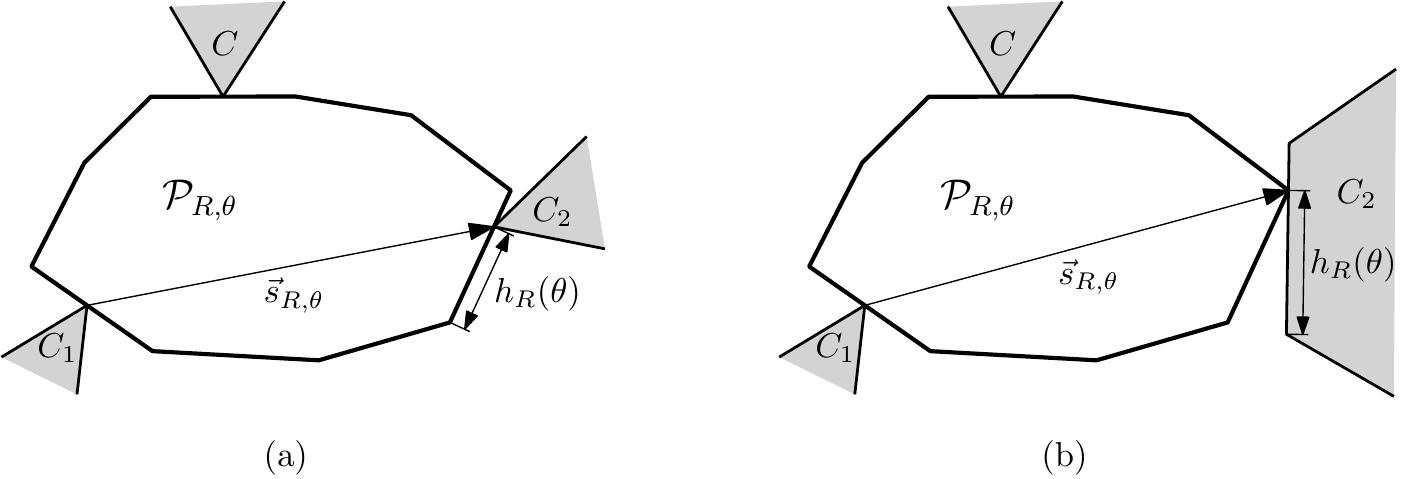}
    \end{center}
    \caption{$\conv_{R,\theta}, \vc{R}{\theta},$ and
      $\len[R]{\theta}$ for a restricted contact pair $R=(C,I)$ and
      $\theta$.  Let $\len[R]{\theta}$ be the distance from the
      clockwise endpoint (with respect to $\vc{R}{\theta}$) of the
      side element of $C_2$ to point element of $C_2$.  (a)
      $\len[R]{\theta}$ when $C_2$ is a corner contact.  (b)
      $\len[R]{\theta}$ when $C_2$ is a side contact.}
    \label{fig:definitions}
  \end{figure}

  For any two fixed contact pairs
  $(C_1,C_2)$, 
  we count the combinatorial changes involving $(C_1,C_2)$, and other
  contact pairs $C$ and $C'$ given in counterclockwise order $C_1, C_2, C,$
  and $C'$ along the boundary of $P$.  See Figure~\ref{fig:combinatorial_change}.
  Note that for $(C_1,C_2)$  with $C_i=(A_i,B_i)$ for $i=1,2$,
  we do not count the combinatorial changes for the pair
  if $A_1=A_2$ or $B_1=B_2$. The combinatorial changes not counted 
  are counted for other two fixed contact pairs.
  
  We use $(C,I)$ to denote a contact pair $C$ restricted to an
  interval $I$ of $\theta$.  Let $\mathcal{R}$ be the set of
  \emph{restricted contact pairs} $(C,I)$ such that
  $C(T)=\{C_1,C_2,C\}$ and $I=I(T)$ for a face $T\in\mathcal{T}$, and
  $C_1,C_2, C$ appear in counterclockwise order along $P$.  For a
  fixed restricted contact pair $R\in\mathcal{R}$ and $\theta\in I$
  for $R=(C,I)$, let $\conv_{R,\theta}$ denote the homothet of
  $P_{\theta}$ which satisfies $C_1,C_2$, and $C$.  We use
  $\vc{R}{\theta}$ to denote the ray from
  the point element of $C_1$ to the point element of $C_2$ in
  $\conv_{R,\theta}$.  Let $\len[R]{\theta}$ be the function that
  denotes the distance from the clockwise endpoint (with respect to
  $\vc{R}{\theta}$) of the side element of $C_2$ to point element of
  $C_2$ with respect to $\conv_{R,\theta}$.  See
  Figure~\ref{fig:definitions}.

  \begin{figure}[ht]
    \begin{center}
      \includegraphics[width=.6\textwidth]{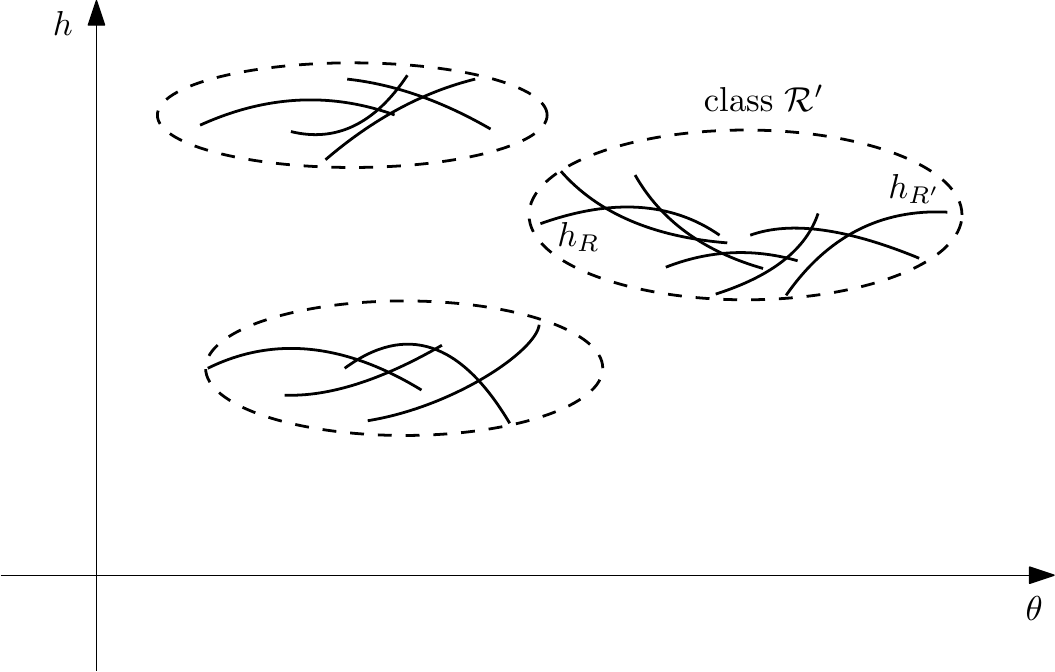}
    \end{center}
    \caption{Partitioning $\mathcal{R}$ by classes using the intersections
      of function graphs of $\env$. Class $\mathcal{R}'$ has $R$ and $R'$, and
      there is a sequence $f_1,\dots,f_j$ with
      $f_1=\lenf{R},f_j=\lenf{R'}$ such that $f_{j'}$ and $f_{j'+1}$
      intersect each other for every $j'=1,\dots,j-1$.}
    \label{fig:grouping}
  \end{figure}

  Observe that $\lenf{R}$ is a partially defined continuous function on $R\in\mathcal{R}$.
  Let $\env=\{\lenf{R}\mid R\in\mathcal{R}\}$.  Two restricted contact
  pairs $R,R'\in\mathcal{R}$ are in the same \emph{class} if and only
  if there is a subset $\{f_1,\dots,f_j\}\subset\env$ with
  $f_1=\lenf{R},f_j=\lenf{R'}$
  such that $f_{j'}$ and $f_{j'+1}$ intersect each other for every
  $j'=1,\dots,j-1$.  Figure~\ref{fig:grouping} illustrates the classes
  of $\mathcal{R}$.

  If a combinatorial change is induced by $C_1,C_2, C,$ and $C'$ at
  $\theta$, we have $\len[R]{\theta}=\len[R']{\theta}$ for distinct restricted contact pairs 
  $R=(C,I)$ and $R'=(C',I')$.
  See Figure~\ref{fig:combinatorial_change}.  Let
  $\mathcal{R'}\subset{\mathcal{R}}$ be a subclass of $\mathcal{R}$
  and let $\env'=\{\lenf{R}\mid R\in\mathcal{R'}\}$.  We verify that
  if $\conv_{R,\theta}$ is feasible, then $\len[R]{\theta}$ appears on the
  lower envelope or upper envelope of $\env'$.

    \begin{figure}[ht]
    \begin{center}
      \includegraphics[width=\textwidth]{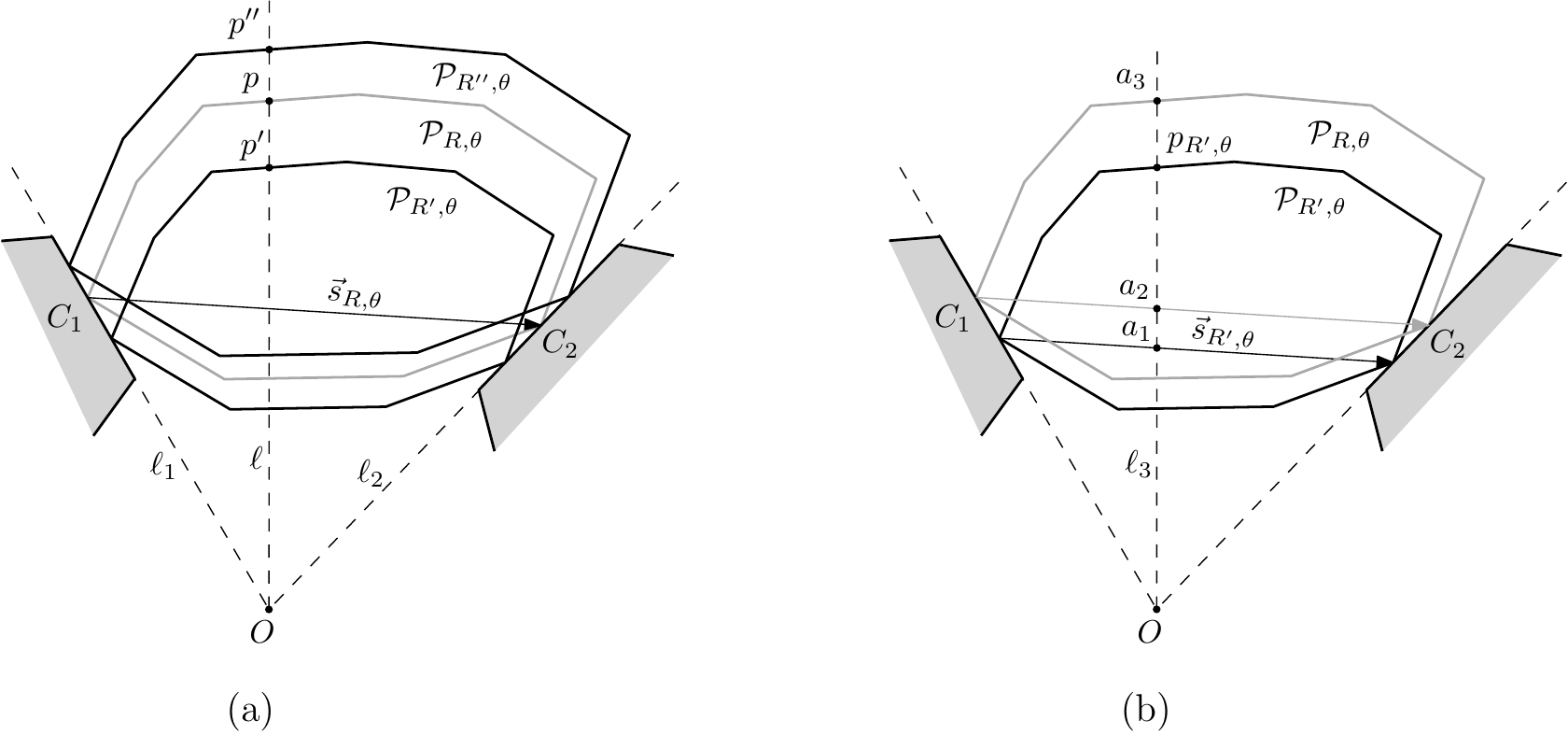}
    \end{center}
    \caption{Proof of Lemma~\ref{lem:class_different}. }
    \label{fig:uniqueness}
  \end{figure}
  
  \begin{lemma}
    \label{lem:class_different}
    Let $R,R',R''\in\mathcal{R}$ be the restricted contact pairs in
    the same class.  If
    $\len[R']{\theta}<\len[R]{\theta}<\len[R'']{\theta}$, then
    $\conv_{R,\theta}$ is not feasible.
  \end{lemma}
   \begin{proof}
   Let $R=(C,I), R'=(C',I')$, and $R''=(C'',I'')$.
    Let $\ell_1$ be the line through the edge element of $C_1$,
    $\ell_2$ be the line through the edge element of $C_2$, and $O$ be
    the intersection of $\ell_1$ and $\ell_2$ unless they are
    parallel.
    For a fixed point $\bar{p}$ on the boundary of $P$, let $p',p,$
    and $p''$ be the points on the boundaries of
    $\conv_{R',\theta}, \conv_{R,\theta},$ and $\conv_{R'',\theta}$,
    respectively, corresponding to $\bar{p}$.  Then $p',p,$ and $p''$
    are on the same line $\ell$ and the order of $p',p,$ and $p''$
    (with $p$ in between $p'$ and $p''$) on $\ell$ remains the same
    for any choice of $\bar{p}$ because
    $\len[R']{\theta}<\len[R]{\theta}<\len[R'']{\theta}$.  Observe
    that $\ell$ passes through $O$ (if exists) and crosses both
    $\vc{R}{\theta}$ and
    $\vc{R''}{\theta}$. 
    See Figure~\ref{fig:uniqueness}(a).  

   Let $\pt{R'}{\theta}$ be the position of the point element of $C'$ in
  $\conv_{R',\theta}$.    
    Let $\ell_3$ be the line through $O$ and $\pt{R'}{\theta}$ if $O$
    exists.  Otherwise, let $\ell_3$ be the line through
    $\pt{R'}{\theta}$ and parallel to $\ell_1$ and $\ell_2$.  Let
    $a_1=\vc{R'}{\theta}\cap \ell_3, a_2=\vc{R}{\theta}\cap \ell_3$
    and let $a_3$ be the point on the boundary of $\conv_{R,\theta}$
    corresponding to $\pt{R'}{\theta}$ on $\conv_{R',\theta}$.
    Without loss of generality, assume that $\ell_3$ is vertical, and
    $\pt{R'}{\theta}$ lies above $a_1$.
    
    Consider the case that
    $\pt{R'}{\theta}$ lies below $a_3$.
    See Figure~\ref{fig:uniqueness}(b).
    We show that $\pt{R'}{\theta}$ does not lie on the segment
    $a_1a_2$.  Since $R',R$ are in the same class, there is a subset
    $\{f_1,\dots,f_j\}\subset\env$ with $f_1=\lenf{R'}, f_j=\lenf{R}$
    such that $f_{j'}$ and $f_{j'+1}$ intersect each other for every    
    $j'=1,\dots,j-1$.  Let $R_{j'}=(C_{j'},I_{j'})$ be the restricted contact pair
    with angle interval $I_{j'}$ such that $\lenf{R_{j'}}=f_{j'}$
    for $j'=1,\dots,j$.  
    Let $\conv'_{R_{j'},\theta'}$ be
    the rotated and scaled copy of the convex hull of $B_1, B_2$, and $B'$
    of $C_{j'}=(A',B')$,
      satisfying $C_1, C_2$ and $C_{j'}$.
    Then $\vc{R_{j'}}{\theta'}$ never intersects $\pt{R'}{\theta}$ for any
    $\theta'\in I_{j'}$ and $j'=1,\dots,j$ since 
    $\pt{R'}{\theta}$ does not lie in the interior of $\conv'_{R_{j'},\theta'}$ but
    $\vc{R_{j'}}{\theta'}$ is contained in the interior of
    $\conv'_{R_{j'},\theta'}$. 
     Let $\theta_{j'}\in I_{j'}$ be one of the
    orientations at which an intersection of $f_{j'}$ and $f_{j'+1}$
    occurs for each
    $j'=1,\dots,j-1$. Then, $\vc{R_{j'}}{\theta_{j'-1}}$ can
    be translated continuously to $\vc{R_{j'}}{\theta_{j'}}$ for
    $\theta$ from $\theta_{j'-1}$ to $\theta_{j'}$ for each
    $j'=2,\dots,j-1$.  Also, $\vc{R'}{\theta}$ can be translated
    continuously to $\vc{R'}{\theta_{1}}$, and $\vc{R}{\theta_{j-1}}$
    can be translated continuously to $\vc{R}{\theta}$.  If
    $\pt{R'}{\theta}$ lies on $a_1a_2$, $\pt{R'}{\theta}$ lies on
    $\vc{R_{j'}}{\theta'}$ for some $j'=1,\dots,j$ and
    $\theta'\in I_{j'}$.  This contradicts to the fact that
    $\vc{R_{j'}}{\theta'}$ never intersects $\pt{R'}{\theta}$.
    Thus, $\pt{R'}{\theta}$ lies above $a_1a_2$.  Then
    $\pt{R'}{\theta}$ lies in the interior of $a_2a_3$, that is,
    $\pt{R'}{\theta}$ lies in the interior of $\conv_{R,\theta}$, and
    therefore $\conv_{R,\theta}$ is not feasible.
    
     Now consider the case that $\pt{R'}{\theta}$ lies above $a_3$.
     Then by an argument similar to the one for the previous case, we can show that
     $\pt{R''}{\theta}$ is contained in the interior of $\conv_{R,\theta}$, where
     $\pt{R''}{\theta}$ is the position of the point element of $C''$ in $\conv_{R'',\theta}$.
     Therefore $\conv_{R,\theta}$ is not feasible.
  \end{proof}

  Let $\mathcal{R'}_{1i}$ be the subset of $\mathcal{R'}$ such that
  $R=(C,I)\in\mathcal{R'}$ belongs to $\mathcal{R'}_{1i}$ if
  $C=(e,v_i)$ is a side contact pair for some edge $e\in Q$.  Let
  $\mathcal{R'}_{2i}$ be the subset of $\mathcal{R'}$ such that
  $R=(C,I)\in\mathcal{R'}$ belongs to $\mathcal{R'}_{2i}$ if
  $C=(v,e_i)$ is a corner contact pair for some vertex $v\in Q$.  Suppose
  that $\left\lvert\mathcal{R'}\right\rvert=m$ and
  $\left\lvert\mathcal{R'}_{1i}\right\rvert=m_{1i},\left\lvert\mathcal{R'}_{2i}\right\rvert=m_{2i}$
  for $i=1,\dots,k$.  Let
  $\env'_{ji}=\{\lenf{R}\mid R\in\mathcal{R'}_{ji}\}$ for $j=1,2$.
  First, we count the breakpoints on the lower envelope of
  $\env'_{1i}$ and the lower envelope $\env'_{2i}$. The number of
  breakpoints on the lower envelope of $\env'_{1i}$ is $O(m_{1i})$
  since $\lenf{R}$ and $\lenf{R'}$ can intersect only at the
  boundaries of their intervals for $R,R'\in\mathcal{R'}_{1i}$.  Let
  $d_i$ be the number of intersections of the function graphs of $\env'_{2i}$,
  and let $d=\sum_{i=1}^{k}d_i$.  Then the number of breakpoints on
  the lower envelope of $\env'_{2i}$ is $O(m_{2i}+d_{i})$.

  \begin{lemma}
    \label{lem:breakpoints_2}
    The number of breakpoints on the lower envelope of $\env'_{1i}$ is
    $O(m_{1i})$ and $\env'_{2i}$ is $O(m_{2i}+d_{i})$ for each
    $i=1,\dots,k$.
  \end{lemma}

  \begin{figure}[ht]
    \begin{center}
      \includegraphics[width=0.35\textwidth]{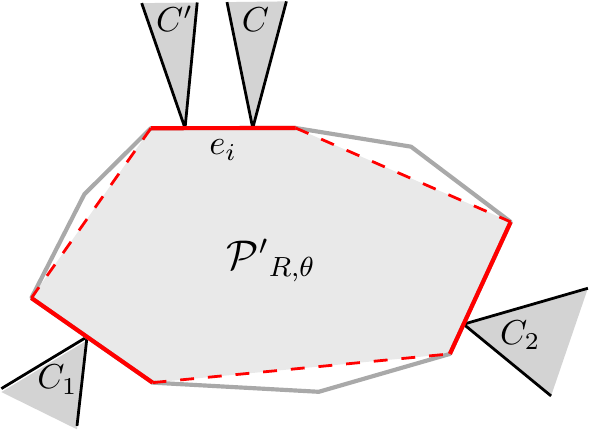}
    \end{center}
    \caption{For $R=(C,I)$ and $R'=(C',I')$ in $\mathcal{R'}_{2i}$,
      and orientation $\theta$ with $\len[R]{\theta}=\len[R']{\theta}$, there exists a
      rotated and scaled copy of the convex hull of $B_1, B_2$, and $e_i$,
      satisfying $C_1, C_2$ and $C, C'$, and feasible.
      The intersection $\len[R]{\theta}=\len[R']{\theta}$ corresponds to 
      a combinatorial change to $\edtt$
      for the convex hull of $B_1,B_2,$ and $e_i$.}
    \label{fig:aligned}
  \end{figure}

  Observe that each intersection of the function graphs of $\env'_{2i}$
  corresponds to a combinatorial change to \edt for the convex hull of
  $B_1,B_2,$ and $e_i$.  See Figure~\ref{fig:aligned}.  By
  Lemma~\ref{lem:class_different}, every combinatorial change appears
  on the lower envelope or upper envelope of $\env'$.  Here, we
  describe the case for the lower envelope of $\env'$.  We count the
  breakpoints of certain types on the lower envelope of $\env'$. 
  The combinatorial changes corresponding to the breakpoints not
  counted here will be counted for other choices of the fixed pair.
  We use \emph{$(a,b)$-change} to denote a combinatorial change
  induced by $a$ side contact pairs and $b$ corner contact pairs.

  \begin{figure}[ht]
    \begin{center}
      \includegraphics[width=0.35\textwidth]{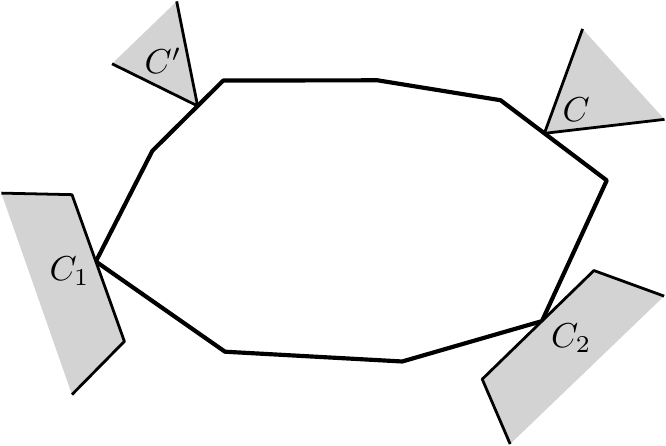}
    \end{center}
    \caption{The $(2,2)$-change induced by $C_1,C_2,C$ and $C'$ is
      counted when $C_2$ and $C$ are chosen as the fixed pair.}
    \label{fig:change_fixed}
  \end{figure}
 
  \myparagraph{Two side contact pairs.}  We count only $(4,0)$-changes
  in this case.  We count $(3,1)$ and $(2,2)$-changes appearing on the
  lower envelope in other choices of the fixed pair.  See
  Figure~\ref{fig:change_fixed}.
  By Corollary~\ref{cor:low_env2} and Lemma~\ref{lem:breakpoints_2},
  the number of breakpoints on the lower envelope of
  $\env'_1=\{f_i\mid i=1,\dots,k\}$ is
  $\sum_{i = 1}^{k} O(m_{1i}\dsn{k}{4}/k)=O(m\dsn{k}{4}/k)$, where
  $f_i$ is the lower envelope of $\env'_{1i}$.  In
  Section~\ref{sec:intersection}, we show that any two continuous
  pieces in $\env'_1$ intersect each other in at most two points.
  Since each $(4,0)$-change corresponds to a breakpoint on the lower
  envelope of $\env'_1$, the number of $(4,0)$-changes is
  $O(m\dsn{k}{4}/k)$.

  \myparagraph{Two corner contact pairs.}  We count only
  $(0,4)$-combinatorial changes in this case. Other changes are
  counted for other choices of the fixed pair.
  By Corollary~\ref{cor:low_env2} and Lemma~\ref{lem:breakpoints_2},
  the number of breakpoints on the lower envelope of
  $\env'_2=\{f_i\mid i=1,\dots,k\}$ is
  $\sum_{i = 1}^{k}
  O((m_{2i}+d_{i})\dsn{k}{4}/k)=O((m+d)\dsn{k}{4}/k)$, where $f_i$ is
  the lower envelope of $\env'_{2i}$. In
  Section~\ref{sec:intersection}, we show that any two continuous
  pieces in $\env'_2$ intersect each other in at most two points.
  Since each $(0,4)$-change corresponds to a breakpoint of the lower
  envelope of $\env'_2$, the number of $(0,4)$-changes is
  $O((m+d)\dsn{k}{4}/k)$.

  \myparagraph{One side contact pair and one corner contact pair.}  We
  count all combinatorial changes other than $(4,0)$-changes and
  $(0,4)$-changes.  First, we count the breakpoints of the lower
  envelope of
  $\env'_{1}=\{\lenf{R}\mid R\in\bigcup_{i=1}^{k}\mathcal{R'}_{1i}\}$
  and of the lower envelope of
  $\env'_{2}=\{\lenf{R}\mid R\in\bigcup_{i=1}^{k}\mathcal{R'}_{2i}\}$,
  then count the breakpoints of the lower envelope of
  $\env'_{1}\cup \env'_{2}$.  In Section~\ref{sec:intersection}, we
  show that any two continuous pieces, both from either $\env'_{1}$ or
  $\env'_{2}$, intersect each other in at most two points.  The number
  of breakpoints of the lower envelope of $\env'_{1}$ can be computed
  in the same way as for counting $(4,0)$-changes, and the result is
  $O(m\dsn{k}{4}/k)$.  The number of breakpoints of the lower envelope
  of $\env'_{2}$ can be computed in the same way as for counting
  $(0,4)$-changes, and the result is $O((m+d)\dsn{k}{4}/k)$.  The
  number of breakpoints of the lower envelope of $\env'$ is
  $O((m+d)\dsn{k}{4}/k)$, since the number of breakpoints of the lower
  envelope of $\env'_{1}$ and of the lower envelope of $\env'_{2}$ is
  $O((m+d)\dsn{k}{4}/k)$, and any two continuous pieces, one from
  $\env'_{1}$ and one from $\env'_{2}$, intersect each other in at
  most four points by Section~\ref{sec:intersection}.  Thus, the
  number the combinatorial changes for the fixed contact pair is
  $O((m+d)\dsn{k}{4}/k)$.  \medskip

  Consider the sum $\sigma$ of the complexity $|\env'|=m$ of $\env'$
  over all classes for a fixed pair. Then the total sum of $\sigma$'s
  for all enumerations of fixed pairs is $O(k^3n^2)$ since
  $|\mathcal{T}|=O(k^3n^2)$.  Similarly, consider the sum $\xi$ of the
  number of intersections ($d$ in the complexities in the previous
  paragraphs) over all classes for a fixed pair.  Then the total sum
  of $\xi$'s for all enumerations of fixed pairs is $O(k^3n^2)$ since
  $\xi$ is bounded by the number of combinatorial changes to $\edtt$
  for the convex hulls of three elements of $P$.  Therefore, we
  conclude with the following theorem.

\begin{theorem} \label{thm:main_bound}
  For a polygonal domain $Q$ of size $n$ and a convex $k$-gon
  $P$, the number of combinatorial changes to the edge Delaunay
  triangulation of $Q$ under $P_\theta$-distance for $\theta$
  increasing from 0 to $2\pi$ is $O(k^2n^2\dsn{k}{4})$.
\end{theorem}
Theorem~\ref{thm:main_bound} directly improves the time complexity 
of the algorithm by Chew and Kedem.
\begin{corollary}
  Given a polygonal domain $Q$ of size $n$ and a convex $k$-gon
  $P$, we can find a largest similar copy of $P$ that can be inscribed
  in $Q$ in $O(k^2n^2\dsn{k}{4}\log{n})$ time using $O(kn)$ space.
\end{corollary}

\section{The number of critical orientations for four contact
  pairs}
\label{sec:intersection}
The orientation $\theta$ at which a combinatorial change to $\edtt$
occurs is called a \emph{critical orientation}.  We consider the
critical orientations $\theta$ at which $\conv_{\theta}$ has contact
with four contact pairs.
We count the number of critical orientations for each $(a,b)$-change
type.  The number of critical orientations of $(0,4)$-change is 2,
which is shown in Appendix B of~\cite{chew1993convex}.  So in this
section, we count the critical orientations for the other
types of combinatorial changes. The following table summarizes the
results.
\begin{center}
  \begin{tabular}{|c||c|c|c|c|c|} \hline Types of $(a,b)$-change &
    $(4,0)$ & $(3,1)$ & $(2,2)$ & $(1,3)$ & $(0,4)$ \\ \hline Number
    of critical orientations & 1 & 2 & 4 & 2 & 2 \\ \hline
  \end{tabular} 
\end{center}

\subsection{Common intersection and directed angles}
We need some technical lemmas before computing the critical
orientations.
For any two lines $\ell$ and $\ell'$ crossing each other in the plane,
the \emph{directed angle}, denoted by $\dangle(\ell,\ell')$, is the angle
from $\ell$ to $\ell'$, measured counterclockwise around their intersection
point.  This definition can be extended for three points $A,O,$ and
$B$, and we use $\dangle{AOB}=\dangle(AO,BO)$ to denote the directed
angle from the line through $AO$ to the line through $BO$, measured
counterclockwise around $O$.  There are some properties of directed
angles which we use in this section.

\begin{proposition}
  \label{prop:directed_angle}
  The following properties hold:
  \begin{enumerate}
  \item For any three points $A,B,$ and $O$, we have
    $\dangle{AOB}=-\dangle{BOA}$ modulo $\pi$.
  \item For any four points $A,B,X,$ and $Y$ such that no three of
    them are collinear, they are concyclic\footnote{A set of points
      are said to be \emph{concyclic} if they lie on a common circle.}
    if and only if $\dangle{XAY}=\dangle{XBY}$ modulo $\pi$.
  \item For any three lines $\ell_1,\ell_2,$ and $\ell_3$ such that no
    two lines are parallel, we have
    $\dangle(\ell_1,\ell_2)+\dangle(\ell_2,\ell_3)=\dangle(\ell_1,\ell_3)$
    modulo $\pi$.  For any four points $A,B,C,$ and $O$, we have
    $\dangle{AOB}+\dangle{BOC}=\dangle{AOC}$ modulo $\pi$.
  \item For any three lines $\ell_1,\ell_2,$ and $\ell_3$ such that no
    two lines are parallel, we have
    $\dangle(\ell_1,\ell_2)+\dangle(\ell_2,\ell_3)+\dangle(\ell_3,\ell_1)=0$
    modulo $\pi$.
  \end{enumerate}
\end{proposition}

\begin{figure}[ht]
  \begin{center}
    \includegraphics[width=.9\textwidth]{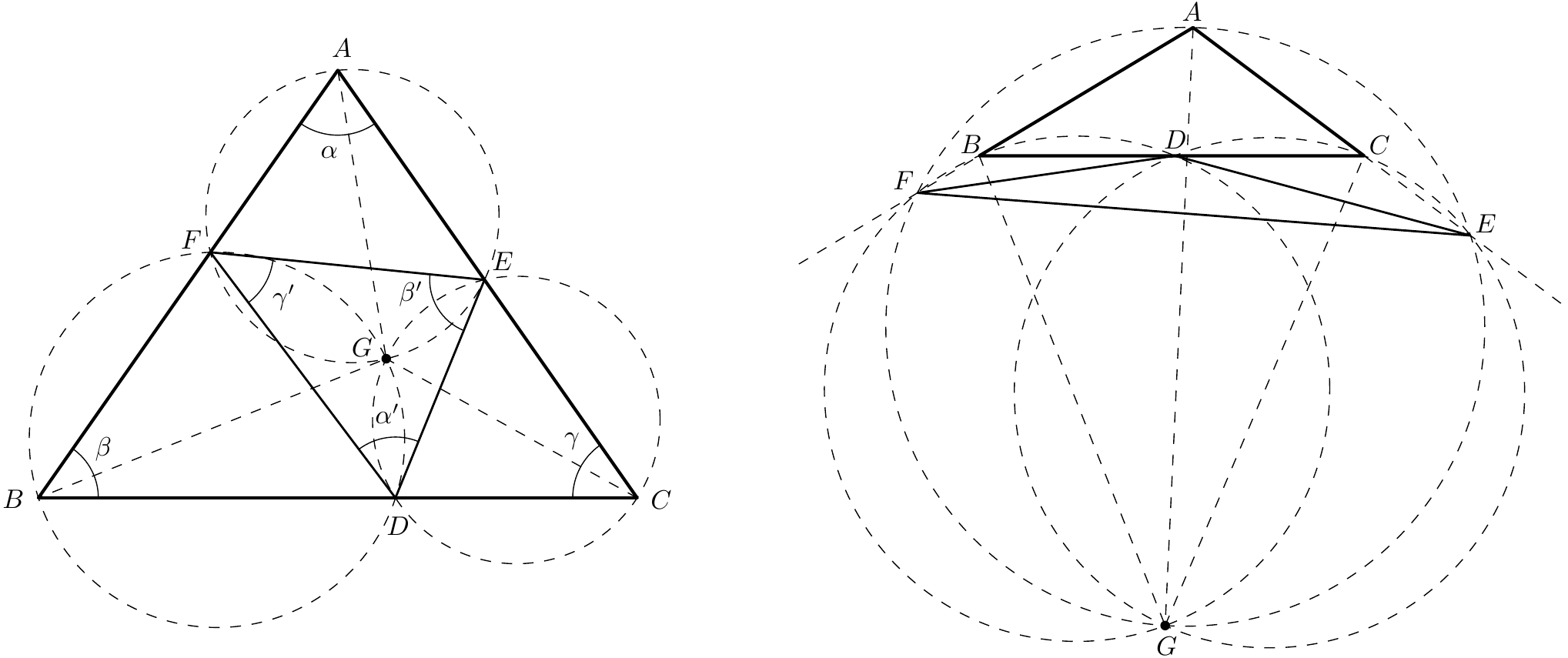}
  \end{center}
  \caption{Some possible configurations of Lemmas~\ref{lem:miquel} and~\ref{lem:triangle}.}
  \label{fig:configuration}
\end{figure}

See Figure~\ref{fig:configuration} for an illustration of Lemmas~\ref{lem:miquel} and~\ref{lem:triangle}.

\begin{lemma}[Miquel's Theorem~\cite{coxeter1967geometry}]
  \label{lem:miquel}
  Let $A, B$, and $C$ be the corners of a triangle, and let $D, E$,
  and $F$ be the points on the lines containing $CB$, $AC$, and $AB$,
  respectively.  Then the three circumcircles to triangles $EAF$,
  $FBD$ and $DCE$ have a common point of intersection.
\end{lemma}

\begin{lemma}
  \label{lem:triangle}
  Let $A, B$, and $C$ be the corners of a triangle, and let $D, E$,
  and $F$ be the points on the lines containing $CB$, $AC$, and $AB$,
  respectively.  Let $G$ be the intersection of three circumcircles to
  triangles $EAF$, $FBD$ and $DCE$, and let
  $\dangle{BAC}=\alpha, \dangle{CBA}=\beta, \dangle{ACB}=\gamma,
  \dangle{EDF}=\alpha', \allowbreak\dangle{FED}=\beta',
  \dangle{DFE}=\gamma'$ modulo $\pi$ for fixed angles $\alpha, \beta, \gamma,
  \alpha', \beta'$, and $\gamma'$.
  \begin{enumerate}
  \item
    $\dangle{AGB}=\gamma-\alpha'-\beta',
    \dangle{BGC}=\alpha-\beta'-\gamma'$, and
    $\dangle{CGA}=\beta-\gamma'-\alpha'$ modulo $\pi$, regardless of
    $\dangle{EFA}$.
  \item If $ABC$ or $DEF$ remains unchanged while
    $\dangle{EFA}$ is changing, $G$ remains unchanged.
  \item
    $\dangle{GDF},\dangle{EDG},\dangle{GED},\dangle{FEG},\dangle{GFE},$
    and $\dangle{DFG}$ remain unchanged while $\dangle{EFA}$ is
    changing.
  \end{enumerate}
\end{lemma}
\begin{proof} See Figure~\ref{fig:configuration} for an illustration.
  $ $
  \begin{enumerate}
  \item Suppose that
    $\dangle{GDF}=\alpha_1,\dangle{EDG}=\alpha_2,\dangle{GED}=\beta_1,\dangle{FEG}=\beta_2,\dangle{GFE}=\gamma_1$,
    and $\dangle{DFG}=\gamma_2$ modulo $\pi$. Then,
    \begin{align*}
      \alpha=\beta_2+\gamma_1, && \beta=\gamma_2+\alpha_1, && \gamma=\alpha_2+\beta_1, && (\text{by Proposition~}\ref{prop:directed_angle}.2) \\
      \alpha'=\alpha_1+\alpha_2, && \beta'=\beta_1+\beta_2, && \gamma'=\gamma_1+\gamma_2
    \end{align*}
    modulo $\pi$. Using the equations above, we have
    \begin{align*}
      \dangle{AGB}&=-\dangle{GBA}-\dangle{BAG}=-\dangle{GBF}-\dangle{FAG}=-\dangle{GDF}-\dangle{FEG} \\
                  &=-\alpha_1-\beta_2=\gamma-\alpha'-\beta'
    \end{align*}
    modulo $\pi$. We can show $\dangle{BGC}=\alpha-\beta'-\gamma'$,
    and $\dangle{CGA}=\beta-\gamma'-\alpha'$ modulo $\pi$ in the same
    way.
  \item Suppose that triangle $ABC$ remain unchanged while
    $\dangle{EFA}$ is changing.  By Lemma~\ref{lem:triangle}.1, the
    circumcircles to triangles $ABG,BCG,$ and $CAG$ remain unchanged.
    Thus, their common intersection $G$ also remains unchanged.  When
    the triangle $DEF$ remains unchanged, $G$ also remains unchanged
    since the circumcircles to triangles $EAF, FBD,$ and $DCE$ remain
    unchanged.
  \item It is deduced by Lemma~\ref{lem:triangle}.2.
  \end{enumerate}
\end{proof}

\subsection{Counting critical orientations}
Now we count the critical orientations which is induced by the same
set of contact pairs. To do this, we count the orientations $\theta$
at which $\homoPt$ satisfies the four contact pairs, $C_i=(A_i,B_i)$
for $i=1,\ldots, 4$. We consider each type of $(a,b)$-changes in the
following.

\begin{figure}[ht]
  \begin{center}
    \includegraphics[width=\textwidth]{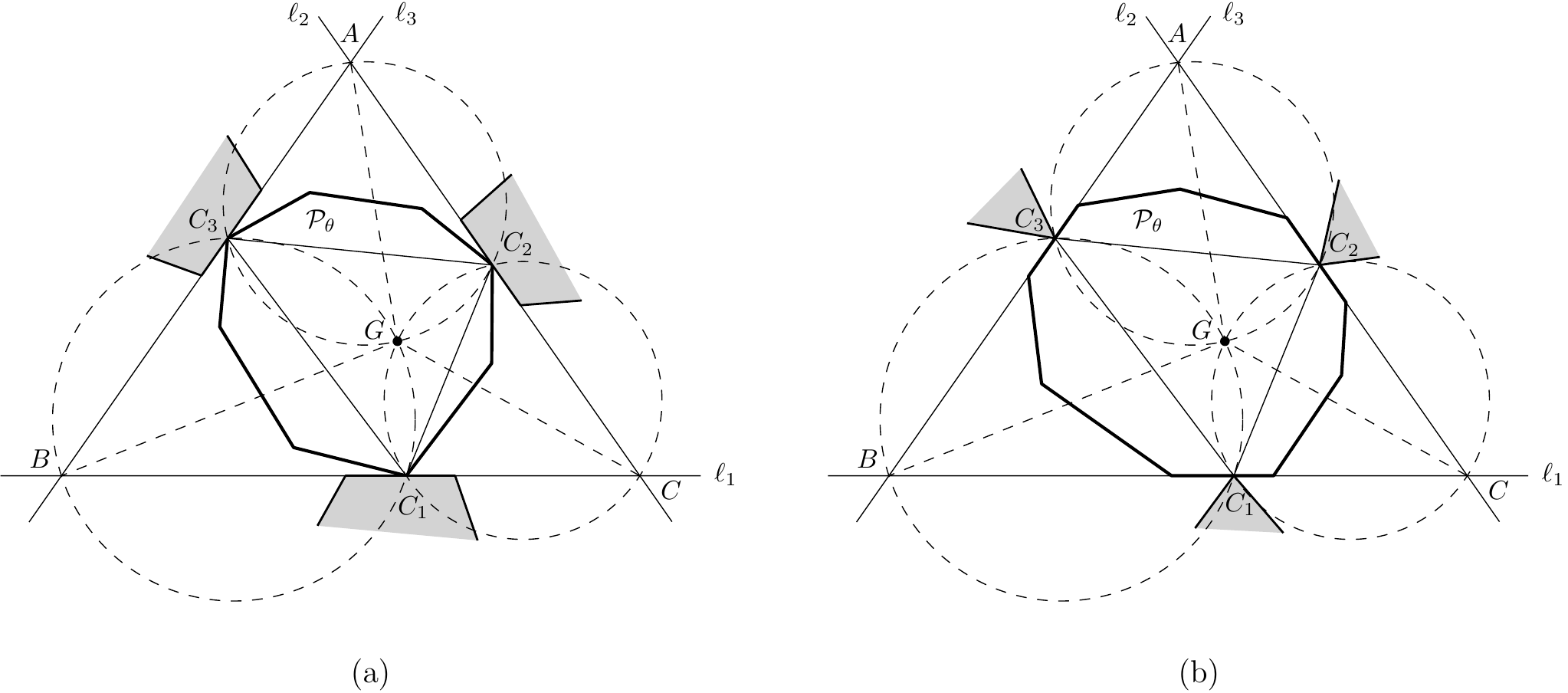}
  \end{center}
  \caption{(a) $\homoPt$ satisfies the side contact pairs
    $C_i=(A_i,B_i)$ for $i=1,2,3$.  (b) $\homoPt$ satisfies the corner
    contact pairs $C_i=(A_i,B_i)$ for $i=1,2,3$. }
  \label{fig:geometric}
\end{figure}

\myparagraph{$(4,0)$-changes and $(3,1)$-changes.}  Suppose that
$C_1,C_2$, and $C_3$ are side contact pairs and $\homoPt$ satisfies
them.  Let $\ell_i$ be the line through
$A_i$ for each $i=1, 2, 3$, and let
$A=\ell_2\cap \ell_3, B=\ell_1\cap \ell_3$, and $C=\ell_1\cap \ell_2$.
By Lemma~\ref{lem:miquel}, the circumcircles to triangles
$AB_2B_3,BB_3B_1,$ and $CB_1B_2$ have a common intersection point,
denoted by $G$.  See Figure~\ref{fig:geometric}(a).  By
Lemma~\ref{lem:triangle}.2, $G$ remains unchanged for any $\theta$
since $A,B,$ and $C$ remain unchanged.  Also, $\dangle{B_2B_1G}$ and
$\dangle{GB_2B_1}$ remain unchanged by Lemma~\ref{lem:triangle}.3.

We use $\ora{XY}$ to denote the vector from point $X$ to point $Y$.
Let $\aff_{s,\vartheta}$ denote the affine transformation with scale
factor $s$ and rotation angle $\vartheta$ in counterclockwise
direction.  Then there exist $s$ and $\vartheta$ such that
$\ora{GB_2}=\aff_{s,\vartheta}\ora{GB_1}$ for every $\homoPt$
satisfying $C_1,C_2$, and $C_3$.  For a vertex $V$ of $\conv$,
$\ora{GV}=t_1\ora{GB_1}+t_2\ora{GB_2}$ for constants $t_1$ and $t_2$.
Therefore, $\ora{GV}=\aff_{s',\vartheta'}\ora{GB_1}$ for some $s'$ and
$\vartheta'$, and the trace of $V$ is a line segment.

If $C_4$ is also a side contact pair, the traces of $B_4$ and $A_4$
intersect each other at most once since the trace of $B_4$ is a line
segment.  Therefore, there is at most one critical orientation for a
$(4,0)$-change.  Consider the case that $C_4$ is a corner contact
pair.  Assume that $B_4=V_1V_2$, for vertices $V_1$ and $V_2$ of $P$.
Since the trace of $V_1$ is a segment, $\ora{GV_1}$ can be
parametrized to $t\vec{v}_1+\vec{c}_1$ for some $\vec{v}_1,\vec{c}_1$
and $t=t(\theta)$.  Since $GV_2=\aff_{s_1,\vartheta_1}\ora{GV_1}$ for
some $s_1$ and $\vartheta_1$, $\ora{GV_2}$ can be parametrized to
$t\vec{v}_2+\vec{c}_2$ for some $\vec{v}_2,\vec{c}_2$ and
$t=t(\theta)$.  If $\homoPt$ satisfies $C_4$, we have
$\ora{GA_4}=t'\ora{GV_1}+(1-t')\ora{GV_2}=t't(\vec{v}_1-\vec{v}_2)+t'(\vec{c}_1-\vec{c}_2)+t\vec{v}_2+\vec{c}_2$
for some $t$ and $t'\in[0,1]$.  Thus, we obtain two equations for $t'$
and $t$ in $x$- and $y$-coordinates.  By eliminating $t'$ using these
two equations, we obtain a quadratic equation for $t$ which has at
most two solutions.  Therefore, there are at most two critical
orientations for a $(3,1)$-change.

\myparagraph{$(1,3)$-changes.}  Suppose that $C_1,C_2$, and $C_3$ are
corner contact pairs and $C_4$ is a side contact pair.  Let $\homoPt$
satisfies these three corner contact pairs.  Let $\ell_i$ be the lines
through $B_i$ for each $i=1, 2, 3$, and let
$A=\ell_2\cap \ell_3,B=\ell_1\cap \ell_3$, and $C=\ell_1\cap \ell_2$.
By Lemma~\ref{lem:miquel}, the circumcircles to triangles
$AA_2A_3,BA_3A_1$, and $CA_1A_2$ have a common intersection point $G$.
See Figure~\ref{fig:geometric}(b).  By Lemma~\ref{lem:triangle}.2, $G$
remains unchanged while $\theta$ increases since $A_1,A_2,$ and $A_3$
remain unchanged.  Also, $\dangle{ABG}$ and $\dangle{GAB}$ remain
unchanged by Lemma~\ref{lem:triangle}.3.  Therefore,
$\ora{GB}=\aff_{s,\vartheta}\ora{GA}$ for some $s$ and $\vartheta$.
Then, $\ora{GB_4}=t_1\ora{GA}+t_2\ora{GB}$ for constants $t_1$ and
$t_2$.  Therefore, $\ora{GB_4}=\aff_{s',\vartheta'}\ora{GA}$ for some
$s'$ and $\vartheta'$, and the trace of $B_4$ is an arc of a circle
through $G$.  Since the traces of $B_4$ and $A_4$ intersect each other
in at most two points, there are at most two critical orientations for
a $(1,3)$-change.

\begin{figure}[ht]
  \begin{center}
    \includegraphics[width=0.4\textwidth]{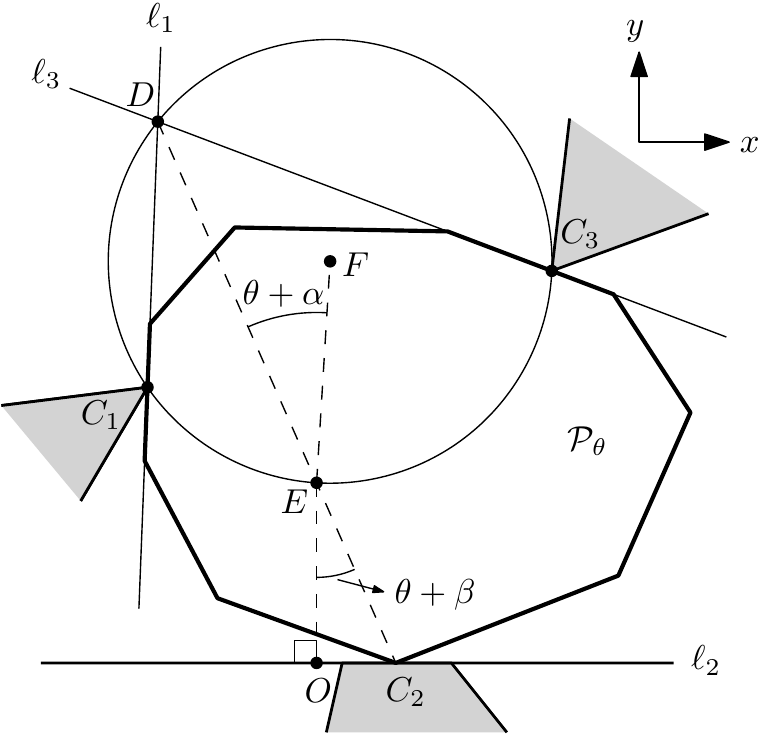}
  \end{center}
  \caption{
    $\homoPt$ satisfies the contact pairs $C_i=(A_i,B_i)$ for
    $i=1,2,3$. }
  \label{fig:22case}
\end{figure}
\myparagraph{$(2,2)$-changes.}  Suppose that $C_1$ and $C_3$ are
corner contact pairs and $C_2$ and $C_4$ are side contact pairs, and
$\homoPt$ satisfies $C_1,C_2$, and $C_3$.  Let $\ell_1, \ell_2,$ and
$\ell_3$ be the lines through $B_1,A_2,$ and $B_3$, respectively, and
let $D=\ell_1\cap \ell_3$.  Let $F$ be the center of the circumcircle
to triangle $A_1A_3D$, and let $E$ be the intersection of the
circumcircle and $B_2D$, other than $D$.  Let $O$ be the perpendicular
foot of $E$ to $\ell_2$.  Let $\theta+\alpha$ be the angle from
$\ora{EF}$ to $\ora{ED}$ and $\theta+\beta$ be the angle from
$\ora{EO}$ to $\ora{EB_2}$ counterclockwise.  See
Figure~\ref{fig:22case}.  Since the angle between $\ora{DA_1}$ and
$\ora{DB_2}$ remains unchanged while $\theta$ increases, $E,F$ and $O$
remain unchanged.  We have $\ora{B_2B_4}=\aff_{s,\vartheta}\ora{B_2D}$
for some $s$ and $\vartheta$.  We also have
$\ora{OB_4}=\ora{OB_2}+\ora{B_2B_4}=\ora{OB_2}+\aff_{s,\vartheta}\ora{B_2D}=
(|EO|\tan(\theta+\beta),0)+\aff_{s,\vartheta}(-|EO|\tan(\theta+\beta)-
2|EF|\cos(\theta+\alpha)\sin(\theta+\beta),|EO|+2|EF|\cos(\theta+\alpha)\cos(\theta+\beta))$.
Observe that each point $V$ on segment $A_4$ can be parametrized to
$\ora{OV}=s\vec{v}+\vec{c}$ for some $\vec{v}$ and $\vec{c}$.  If
$B_4$ touches $A_4$, then $s\vec{v}+\vec{c}=\ora{OB_4}$
is satisfied for some $s$. Thus, we obtain two equations for $\theta$
and $s$ in $x$- and $y$-coordinates.  By eliminating $s$ using these
two equations, we obtain the equation for $\theta$.  By substituting
$t=\tan{\frac{\theta+\beta}{2}}$, we obtain a quartic equation for $t$
which has at most four solutions.  Therefore, there are at most four
critical orientations for a $(2,2)$-change.

\bibliography{reference.bib} \bibliographystyle{plainurl}

\begin{thebibliography}{10}

\bibitem{AAS98}
Pankaj~K Agarwal, Nina Amenta, and Micha Sharir.
\newblock Largest placement of one convex polygon inside another.
\newblock {\em Discrete \& Computational Geometry}, 19(1):95--104, 1998.

\bibitem{agarwal1999motion}
Pankaj~K Agarwal, Boris Aronov, and Micha Sharir.
\newblock Motion planning for a convex polygon in a polygonal environment.
\newblock {\em Discrete \& Computational Geometry}, 22(2):201--221, 1999.

\bibitem{atallah1983dynamic}
Mikhail~J Atallah.
\newblock Dynamic computational geometry.
\newblock In {\em Proceedings of the 24th Annual Symposium on Foundations of
  Computer Science (FOCS 1983)}, pages 92--99. IEEE, 1983.

\bibitem{AvnaimBoissonnat88}
Francis Avnaim and Jean~Daniel Boissonnat.
\newblock Polygon placement under translation and rotation.
\newblock In Robert Cori and Martin Wirsing, editors, {\em STACS 88}, pages
  322--333. Springer Berlin Heidelberg, 1988.

\bibitem{bae2020empty}
Sang~Won Bae and Sang~Duk Yoon.
\newblock Empty squares in arbitrary orientation among points.
\newblock In {\em Proceedings of the 36th International Symposium on
  Computational Geometry (SoCG 2020)}. Schloss Dagstuhl-Leibniz-Zentrum f{\"u}r
  Informatik, 2020.

\bibitem{Chazelle1983}
Bernard Chazelle.
\newblock The polygon containment problem.
\newblock In F.P. Preparata, editor, {\em Advances in Computing Research, Vol
  I: Computational Geometry}, pages 1--33. JAI Press Inc., 1983.

\bibitem{chew1989placing}
L~Paul Chew and Klara Kedem.
\newblock Placing the largest similar copy of a convex polygon among polygonal
  obstacles.
\newblock In {\em Proceedings of the 5th Annual Symposium on Computational
  Geometry (SoCG 1989)}, pages 167--173, 1989.

\bibitem{chew1993convex}
L~Paul Chew and Klara Kedem.
\newblock A convex polygon among polygonal obstacles: Placement and
  high-clearance motion.
\newblock {\em Computational Geometry}, 3(2):59--89, 1993.

\bibitem{coxeter1967geometry}
Harold Scott~Macdonald Coxeter and Samuel~L Greitzer.
\newblock {\em Geometry revisited}, volume~19.
\newblock The Mathematical Association of America, 1967.

\bibitem{Fleischer1992}
Rudolf Fleischer, Kurt Mehlhorn, G{\"u}nter Rote, Emo Welzl, and Chee Yap.
\newblock Simultaneous inner and outer approximation of shapes.
\newblock {\em Algorithmica}, 8(1):365, 1992.

\bibitem{fortune1985fast}
Steven Fortune.
\newblock A fast algorithm for polygon containment by translation.
\newblock In {\em Proceedings of the 12th International Colloquium on Automata,
  Languages, and Programming (ICALP 1985)}, pages 189--198. Springer, 1985.

\bibitem{DCGhandbook}
Jacob~E. Goodman, Joseph O'Rourke, and Csaba~D. T\'oth, editors.
\newblock {\em Handbook of Discrete and Computational Geometry}.
\newblock CRC Press LLC, 3rd edition, 2017.

\bibitem{huttenlocher1992dynamic}
Daniel~P Huttenlocher, Klara Kedem, and Jon~M Kleinberg.
\newblock On dynamic {V}oronoi diagrams and the minimum {H}ausdorff distance
  for point sets under {E}uclidean motion in the plane.
\newblock In {\em Proceedings of the 8th Annual Symposium on Computational
  Geometry (SoCG 1992)}, pages 110--119, 1992.

\bibitem{KedemSharir1990}
Klara Kedem and Micha Sharir.
\newblock An efficient motion-planning algorithm for a convex polygonal object
  in two-dimensional polygonal space.
\newblock {\em Discrete \& Computational Geometry}, 5:43--75, 1990.

\bibitem{lee2020largest}
Seungjun Lee, Taekang Eom, and Hee-Kap Ahn.
\newblock Largest triangles in a polygon.
\newblock {\em arXiv preprint arXiv:2007.12330}, 2020.

\bibitem{LevenSharir87}
Daniel Leven and Micha Sharir.
\newblock Planning a purely translational motion for a convex object in
  two-dimensional space using generalized voronoi diagrams.
\newblock {\em Discrete \& Computational Geometry}, 2:9--31, 1987.

\bibitem{megiddo1983applying}
Nimrod Megiddo.
\newblock Applying parallel computation algorithms in the design of serial
  algorithms.
\newblock {\em Journal of the ACM}, 30(4):852--865, 1983.

\bibitem{o1985retraction}
Colm {\'O}'D{\'u}nlaing and Chee~K Yap.
\newblock A retraction method for planning the motion of a disc.
\newblock {\em Journal of Algorithms}, 6(1):104--111, 1985.

\bibitem{sharir1994extremal}
Micha Sharir and Sivan Toledo.
\newblock Extremal polygon containment problems.
\newblock {\em Computational Geometry}, 4(2):99--118, 1994.

\bibitem{szemeredi1974problem}
Endre Szemer{\'e}di.
\newblock On a problem by davenport and schinzel.
\newblock {\em Acta Arithmetica}, 25:213--224, 1974.

\bibitem{Toledo91}
Sivan Toledo.
\newblock Extremal polygon containment problems.
\newblock In {\em Proceedings of the 7th Annual Symposium on Computational
  Geometry (SoCG 1991)}, pages 176--185. {ACM}, 1991.

\end{thebibliography}
\end{document}